\documentclass[11pt]{article}
\pdfoutput=1
\usepackage[nosort]{cite}

\usepackage{xcolor}

\usepackage{epsfig}
\usepackage{amsfonts}
\usepackage{amscd}
\usepackage{latexsym}
\usepackage{amsmath,amssymb}
\usepackage{amsthm}
\usepackage{verbatim}
\usepackage{setspace}
\usepackage{color}
\usepackage{fancyhdr}
\usepackage{tikz}
\usepackage{tikz-cd}
\usepackage{tikz-3dplot}
\usepackage{slashed}
\usepackage{soul}
\usepackage{multirow}

\usetikzlibrary{calc}
\usetikzlibrary{topaths}
\usetikzlibrary{decorations}
\usetikzlibrary{decorations.pathmorphing}
\usepackage{nicematrix}

\usetikzlibrary{zx-calculus}

\usepackage{draft}
\usepackage{hyperref}
\usepackage{graphicx,color,subfig}\usepackage[font=small]{caption}
\usepackage{cite}
\usepackage{mciteplus}
\usepackage{skak}
\usepackage{centernot}

\usepackage{enumitem}
\usepackage{crossreftools}

\newtheorem{theorem}[equation]{Theorem}

\newtheorem{proposition}[equation]{Proposition}
\newtheorem{claim}[equation]{Claim}

\numberwithin{equation}{section}

\newcommand{\rank}{\operatorname{rank}}
\newcommand{\im}{\operatorname{im}}

\newcommand{\sm}{\smallsetminus}
\newcommand{\mdeg}{\operatorname{max-deg}}
\newcommand{\mldeg}{\operatorname{max-left-deg}}
\newcommand{\mrdeg}{\operatorname{max-right-deg}}
\newcommand{\mord}{\operatorname{max-ord}}

\newcommand{\Plus}{\mathord{\tikz[scale=1.5,baseline=1pt]\draw[line width=0.2ex, x=1ex, y=1ex] (0.5,0) -- (0.5,1)(0,0.5) -- (1,0.5);}}
\newcommand{\Zero}{{\mathsf 0}}
\newcommand{\One}{{\mathsf 1}}

\def\({\left(}
\def\){\right)}

\def\>{\rangle}
\def\<{\langle}

\begin{document}

\begin{titlepage}
	

\hfill \\
\hfill \\
\vskip 1cm

\title{Deformed LDPC codes with spontaneously broken non-invertible duality symmetries}

\author{Pranay Gorantla and Tzu-Chen Huang}

\address{Leinweber Institute for Theoretical Physics \& Enrico Fermi Institute, University of Chicago}


\vspace{2.0cm}

\begin{abstract}\noindent
Low-density parity check (LDPC) codes are a well known class of Pauli stabiliser Hamiltonians that furnish fixed-point realisations of nontrivial gapped phases such as symmetry breaking and topologically ordered (including fracton) phases. In this work, we propose symmetry-preserving deformations of these models, in the presence of a transverse field, and identify special points along the deformations with interesting features: (i) the special point is frustration-free, (ii) its ground states include a product state and the code space of the underlying code, and (iii) it remains gapped in the thermodynamic (infinite volume) limit. So the special point realises a first-order transition between (or the coexistence of) the trivial gapped phase and the nontrivial gapped phase associated with the code. In addition, if the original model has a non-invertible duality symmetry, then so does the deformed model. In this case, the duality symmetry is spontaneously broken at the special point, consistent with the associated anomaly.

A key step in proving the gap is a coarse-graining/blocking procedure on the Tanner graph of the code that allows us to apply the martingale method successfully. Our model, therefore, provides the first application of the martingale method to a frustration-free model, that is not commuting projector, defined on an arbitrary Tanner graph.

We also discuss several familiar examples on Euclidean spatial lattice. Of particular interest is the 2+1d transverse field Ising model: while there is no non-invertible duality symmetry in this case, our results, together with known numerical results, suggest the existence of a tricritical point in the phase diagram.
\end{abstract} 

\vfill
	
\end{titlepage}

\eject

\tableofcontents

\section{Introduction}

The connection between error correcting codes and gapped phases of matter has received a lot of attention in the recent years. A concrete demonstration of this connection is provided by \emph{low-density parity check} (LDPC) codes: they are local Pauli stabiliser Hamiltonians\footnote{A Pauli stabiliser is a product of Pauli $Z$/$X$ operators. A stabiliser is said to be $Z$-type ($X$-type) if it is a product of only Pauli $Z$ (Pauli $X$) operators. A local Pauli stabiliser Hamiltonian is a local Hamiltonian where each local term is a Pauli stabiliser, and any two local terms commute with each other.} that are exactly solvable and furnish fixed-point realisations of a variety of gapped phases of matter. Typically, they fall into two categories: \emph{classical LDPC codes}, where all the stabilisers are of $Z$-type, such as the 1+1d Ising model, and \emph{quantum LDPC codes}, where stabilisers need not all be of $Z$-type, such as the 2+1d Toric Code and the X-Cube model. The former describe symmetry breaking phases, whereas the latter capture topologically ordered phases (including fracton phases). For a thorough review of this subject from the condensed matter physics point of view, see \cite{Rakovszky:2023fng,Rakovszky:2024iks}. (See also \cite{Breuckmann:2021yvk} for a complementary perspective from the quantum information side.)

Meanwhile, another development has made a huge impact on our understanding of phases of matter: \emph{generalised symmetries} \cite{Gaiotto:2014kfa}. For instance, phase transitions that were once thought to be beyond the Landau-Ginzburg paradigm are now reinterpreted as spontaneous breaking of \emph{higher-form symmetries}, i.e., group-like symmetries generated by unitary operators supported on submanifolds with codimension larger than $1$. Another class of symmetries that has received significant attention is \emph{non-invertible symmetries}, i.e., symmetries whose generators lack inverses (and so are far from being group-like). See \cite{McGreevy:2022oyu,Cordova:2022ruw,Brennan:2023mmt,Bhardwaj:2023kri,Schafer-Nameki:2023jdn, Luo:2023ive,Shao:2023gho,Carqueville:2023jhb} for a review of these developments.

While they were originally studied in the context of continuum quantum field theory, there has been a growing interest in realising non-invertible symmetries exactly on the lattice, where new features, such as mixing with lattice translations, arise. A prototypical example is the Kramers-Wannier (KW) duality symmetry of the 1+1d critical Ising model \cite{Grimm:1992ni,Oshikawa:1996dj,Ho:2014vla,Hauru:2015abi,Aasen:2016dop,Seiberg:2023cdc,Seiberg:2024gek}. (See \cite{Inamura:2021szw,Tan:2022vaz,Eck:2023gic,Mitra:2023xdo,Sinha:2023hum,Fechisin:2023odt,Yan:2024eqz,Okada:2024qmk,Seifnashri:2024dsd,Bhardwaj:2024wlr,Chatterjee:2024ych,Bhardwaj:2024kvy,Khan:2024lyf,Jia:2024bng,Lu:2024ytl,Li:2024fhy,Ando:2024hun,ODea:2024tkt,Pace:2024tgk} for more 1+1d examples and \cite{Delcamp:2023kew,Inamura:2023qzl,Moradi:2023dan,Cao:2023doz, Cao:2023rrb,ParayilMana:2024txy,Spieler:2024fby,Choi:2024rjm,Hsin:2024aqb,Cao:2024qjj} for 2+1d examples.) Concretely, there are qubits on the sites of the 1d chain, and they interact according to the Hamiltonian
\ie
H_\text{crit-Is} = - \sum_i (Z_i Z_{i+1} + X_i)~.
\fe
The KW duality operator acts as
\ie
\mathsf D_\mathrm{KW}~:~ X_i \to Z_i Z_{i+1} \to X_{i+1}~.
\fe
It is easy to see that the Hamiltonian is invariant under this action, and moreover, applying $\mathsf D_\mathrm{KW}$ twice is equivalent to the lattice translation by one unit (within the $\mathbb Z_2$-symmetric subspace) \cite{Seiberg:2023cdc,Seiberg:2024gek}. The latter fact is codified in the non-invertible fusion rule
\ie
\mathsf D_\mathrm{KW}^2 = T (1+\eta)~,
\fe
where $\eta = \prod_i X_i$ is the $\mathbb Z_2$ symmetry operator.

In a recent work \cite{Gorantla:2024ocs}, we generalised this example to non-invertible duality symmetries in a large class of models based on LDPC codes in the presence of transverse field, including the Wegner duality symmetry of the 3+1d self-dual lattice $\mathbb Z_2$ gauge theory.\footnote{The topological defect associated with the Wegner duality symmetry was constructed in the Euclidean lattice gauge theory in \cite{Koide:2021zxj}, and its continuum counterpart was found in \cite{Choi:2021kmx,Kaidi:2021xfk}.} In particular, we gave a tensor network representation of the non-invertible duality operator based on the Tanner graph of the code and analysed its fusion algebra using a powerful graphical tool from quantum information known as \emph{ZX-calculus} \cite{Coecke:2008lcg,Duncan:2009ocf} (see \cite{vandeWetering:2020giq} for a review). Our work provided a unifying framework to study non-invertible duality symmetries in any dimension, and also beyond Euclidean lattices.

The presence of non-invertible symmetries exactly on the lattice imposes powerful constraints on the low energy physics---e.g., in the case of KW duality symmetry, it has been shown that the low energy phase is either gapless or gapped with a ground state degeneracy that is a multiple of $3$ \cite{Levin:2019ifu,Seiberg:2024gek}. Such constraints do not allow for a trivially gapped phase, but they always allow a gapped phase where the non-invertible duality symmetry is spontaneously broken. While the 1+1d critical Ising model does not realise a spontaneously broken KW duality symmetry, it is realised at a \emph{frustration-free} point\footnote{A frustration-free Hamiltonian is a local Hamiltonian such that (i) each local term is positive semi-definite and (ii) there are zero energy ground states. The first property implies that all energies are nonnegative, and the second property implies that the zero energy states (i.e., the ground states) must be annihilated by each and every local term of the Hamiltonian.} along a KW duality-preserving deformation of this model \cite{OBrien:2017wmx}:\footnote{See also \cite{Sannomiya:2017foz} for a similar deformation, and see \cite{2015PhRvL.115p6401R,2015PhRvB..92w5123R,Seiberg:2024gek} for other deformations preserving the $\mathbb Z_2$ and the Kramers-Wannier duality symmetries.}
\ie\label{obrien-fendley-deform}
H^\lambda_\text{crit-Is} = - \sum_i \left[Z_i Z_{i+1} + X_i - \frac\lambda2 (X_{i-1} Z_i Z_{i+1} + Z_{i-1} Z_i X_{i+1})\right]~.
\fe
The frustration-free point is at $\lambda = 1$, where the ground states are $|{+}\cdots{+}\>$, $|0\cdots0\>$, and $|1\cdots1\>$.\footnote{Here, $|0\>$ and $|1\>$ are the eigenstates of $Z$ with eigenvalues of $1$ and $-1$, and $|{+}\>$ and $|{-}\>$ are the eigenstates of $X$ with eigenvalues $1$ and $-1$, respectively.} An observant reader might notice that the first state is the ground state of the trivial paramagnet without the Ising term, whereas the last two are the ground states of the Ising model without the transverse field term (i.e., the classical LDPC code associated with the Ising model). Indeed, the KW duality symmetry exchanges the former with the symmetric superposition of the latter.

In another recent work \cite{Gorantla:2024ptu}, we proposed a similar Wegner duality-preserving deformation of the 3+1d self-dual lattice $\mathbb Z_2$ gauge theory. Along this deformation, we identified a frustration-free point where the Wegner duality symmetry is spontaneously broken. This model features (i) nine exactly solvable and exactly degenerate ground states (a trivial product state and eight topologically ordered Toric Code ground states), and (ii) a gap that remains nonzero in the thermodynamic limit. We emphasise that this model is not commuting projector, so the fact that the ground states remain gapped in the thermodynamic limit is highly nontrivial, even more so given the model is in 3+1d.

It seems only natural to expect that the above deformations should extend to models based on LDPC codes in the presence of transverse field. In this work, we confirm this expectation, i.e., we propose a symmetry-preserving deformation of an arbitrary classical LDPC code, as well as a quantum \emph{Calderbank–Shor–Steane} (CSS) code,\footnote{A quantum CSS code is a special kind of quantum LDPC code where each stabiliser is a product of either Pauli $Z$ operators or Pauli $X$ operators, but not both. Examples include 2+1d Toric Code, X-Cube model, etc.} in transverse field. Furthermore, we identify a special point along the deformation with the following features:
\begin{enumerate}
\item the Hamiltonian is frustration-free, but not commuting projector,
\item the ground states are exactly solvable and exactly degenerate---one of them is a trivial product state and the rest form the \emph{code space} of the original code\footnote{The code space of an LDPC code is the space of ground states of the associated Hamiltonian.}---and
\item they remain gapped in the thermodynamic (infinite volume) limit.
\end{enumerate}
In other words, the special point realises a first-order transition between the trivial gapped phase and the nontrivial gapped phase described by the original code. If, moreover, the original model has a non-invertible duality symmetry, then the deformation preserves this symmetry, and it is spontaneously broken at the special point.

The main technical achievement of this work is to show that the Hamiltonian at the frustration-free point is gapped in the thermodynamic limit. Although generalising our proposal to LDPC codes provides a unifying framework to analyse such duality-preserving deformations in any dimension, and also beyond Euclidean lattices, proving gaps of local Hamiltonians at this level of generality is hard, if not impossible \cite{Cubitt:2015lta}. To the best of our knowledge, commuting projector models, such as LDPC codes themselves, are the only models that can be defined this generally and can be shown to be gapped (for obvious reasons). Our deformed model, therefore, provides a first example of a non-commuting projector model that can be defined on arbitrary lattices, Euclidean or otherwise, and can be shown to be gapped.

The breakthrough that led to this achievement is as follows. There are a couple of well known methods to prove the gap of a frustration-free Hamiltonian: Knabe's method \cite{Knabe:1988} and the martingale method \cite{Fannes:1992}. These methods have found numerous important applications in diverse dimensions \cite{Nachtergaele:1996,Verstraete:2006mdr,Hastings:2005pr,Bravyi:2015ksj,Gosset:2016,Kastoryano_2018,Lemm:2019pxa,Abdul-Rahman:2019wna,Anshu:2019mhi,Lemm:2019jxn,Pomata:2019uap,Andrei:2022yym,Young_2024,Lemm:2024yre,Rai:2024rge} (see also \cite{Cirac:2020obd} for a related review). A key step in both methods is \emph{coarse-graining/blocking} the lattice into large, but finite, and overlapping blocks. Crucially, the overlap between adjacent blocks must be large enough, but the blocks themselves cannot be too large, and the number of blocks that overlap simultaneously cannot be too large. It is clear that there is delicate balance in choosing the blocks appropriately, and this balance can be achieved on Euclidean lattices. To apply these methods to models defined on arbitrary (bounded-degree, bipartite) graphs, one has to come up with a suitable coarse-graining/blocking procedure on the graph. This seems impossible at first, but this is precisely what we achieve in this work, which then allows us to prove the gap of the deformed model. It is natural to ask if our procedure can be used to obtain finite-size criteria on the gaps of frustration-free Hamiltonians beyond Euclidean lattice, similar to the ones on Euclidean lattice \cite{Gosset:2016,Lemm:2019pxa,Anshu:2019mhi}.

The rest of this work is organised as follows. In Section \ref{sec:graph}, we briefly discuss some useful notions of graph theory. Readers who wish to jump to the physics can skip this section without much difficulty. In Section \ref{sec:gTFIM}, we review the quantum Hamiltonian associated to a classical LDPC code in the presence of transverse field. We review the invertible and non-invertible symmetries, and discuss the ground states at special points in the phase diagram. In Section \ref{sec:graph-deformed}, which is the meat of this work, we introduce a symmetry-preserving deformation of this Hamiltonian and identify a special point where the Hamiltonian is frustration-free. We find that the ground states at this point include a product state and the code space of the code, and show that there are no other ground states. We then review the martingale method of proving the gap of a frustration-free Hamiltonian, point out and resolve some subtleties in using this method on a graph, and then apply this method to prove that the deformed model is gapped at the frustration-free point. We end this section with deformations of some familiar examples, such as the 1+1d and 2+1d Ising models, and the 2+1d plaquette Ising model. In Section \ref{sec:css}, we extend our proposed deformation to quantum CSS codes in the presence of transverse field. We then discuss the deformations in various well known examples, such as the 2+1d and 3+1d Toric Codes, and the X-Cube model. In Section \ref{sec:discuss}, we summarise our results and comment on the phase diagram of the deformed model and its generalisation to $\mathbb Z_N$ qudits. In particular, we combine our purely analytical results with known numerical data to sketch the plausible phase diagram in the translation-invariant deformed model on a Euclidean lattice in any dimension. We then end this section with some potential future directions.

There are three appendices containing some technical details. In Appendix \ref{app:eq-wt}, we provide the proof of a claim made in Section \ref{sec:graph-proofofGSD}. This is the main combinatorial component of this work that underlies many other results concerning the deformed model. In Appendix \ref{app:coarse-grain}, we spell out a general coarse-graining procedure on arbitrary hypergraphs associated to local (frustration-free) Hamiltonians. This is the key step in the martingale method that allows us to prove the gap of the deformed model. In Appendix \ref{sec:graph-delta}, we prove the upper bound on the martingale function, a quantity that controls the gap of the Hamiltonian in the thermodynamic limit, for the deformed model.

\section{Basic notions of graphs}\label{sec:graph}
In this section, we discuss some useful notions in graph theory, which can be found in any standard textbook on graph theory, such as \cite{Diestel}.

A \emph{graph} $\mathcal G$ is given by a set of \emph{vertices} $\mathcal V$ and a set of \emph{edges} $\mathcal E \subseteq \binom{\mathcal V}2$ between the vertices.\footnote{The notation $\binom{\mathcal V}k$ denotes the collection of all $k$-element subsets of $\mathcal V$. This is analogous to the notation $2^{\mathcal V}$, which denotes the power set of $\mathcal V$, i.e., the collection of all subsets of $\mathcal V$. Indeed, $2^{\mathcal V} = \bigsqcup_{k=0}^{|\mathcal V|} \binom{\mathcal V}k$, which resembles the identity $2^n = \sum_{k=0}^n \binom nk$. Here and below, we write $A\sqcup B$ to denote the union of two disjoint sets (a.k.a. the \emph{disjoint union} of) $A$ and $B$.}
For simplicity, we write $x\sim y$ if and only if there is an edge between $x$ and $y$, i.e., $x\sim y \iff \{x,y\} \in \mathcal E$.

The \emph{adjacency matrix} of $\mathcal G$ is a $|\mathcal V|\times |\mathcal V|$ matrix $\mathfrak a$ with entries $\mathfrak a_{xy} = 1$ if $x\sim y$ and $0$ otherwise. The \emph{degree} of a vertex $x$, denoted as $\deg(x)$, is the number of edges incident to it, i.e., $\deg(x) := \sum_{y\in \mathcal V} \mathfrak a_{xy}$. Let $\mdeg(\mathcal G) := \max_{x\in \mathcal V} \deg(x)$ be the maximum among all the degrees. We say $\mathcal G$ is of \emph{bounded-degree} if there is a positive integer $\Delta$ such that $\mdeg(\mathcal G) \le \Delta$.

Given two vertices $x$ and $y$, a \emph{path} between them is a sequence of distinct vertices $x = x_1,x_2,\ldots,x_n,x_{n+1} = y$ such that $x_i \sim x_{i+1}$ for $i=1,\ldots,n$. We say $\mathcal G$ is \emph{connected} if there is a path between any two vertices.

The \emph{length} of a path is the number of edges in the path. A \emph{shortest path} between $x$ and $y$ is a path with smallest possible length between those vertices, and the length of the shortest path is known as the \emph{distance}, denoted as $d(x,y)$. Given a vertex $x$ and a nonnegative integer $r$, the \emph{ball} of radius $r$ centred at $x$ is $B(x,r) := \{ y\in \mathcal V: d(x,y) \le r\}$, and the \emph{sphere} of radius $r$ centred at $x$ is $S(x,r) := \{y \in \mathcal V: d(x,y) = r\}$. In particular, $B(x,0) = S(x,0) =\{x\}$, and $\mathscr N(x) := S(x,1)$ denotes the set of vertices adjacent to $x$, i.e., the \emph{neighbourhood} of $x$.

An \emph{automorphism} $\phi : \mathcal V \to \mathcal V$ of a graph is a bijective map on its vertices such that $\phi(x) \sim \phi(y) \iff x \sim y$ for all $x,y\in \mathcal V$. In particular, any automorphism preserves the degrees, i.e., $\deg(\phi(x)) = \deg(x)$, and the distances, i.e., $d(\phi(x),\phi(y)) = d(x,y)$. It follows that $|B(\phi(x),r)| = |B(x,r)|$, or equivalently, $|S(\phi(x),r)| = |S(x,r)|$ for any $r\ge 0$.

A \emph{subgraph} $\mathcal G'$ of a graph $\mathcal G$ is a subset of vertices $\mathcal V' \subset \mathcal V$ and edges $\mathcal E'\subset \mathcal E$ such that the ends of edges in $\mathcal E'$ are contained in $\mathcal V'$. The subgraph is said to be \emph{induced} if all the edges in $\mathcal E$ between the vertices of $\mathcal V'$ are present in $\mathcal E'$, i.e., for any $x,y\in \mathcal V'$, $\{x,y\}\in \mathcal E\implies \{x,y\}\in\mathcal E'$. It is said to be \emph{connected} if any two vertices in $\mathcal V'$ are connected by a path within $\mathcal G'$.

Given two subgraphs $\mathcal G_1$ and $\mathcal G_2$, their \emph{intersection} $\mathcal G_1 \cap \mathcal G_2$ is defined as the subgraph on vertices $\mathcal V_1 \cap \mathcal V_2$ with edges $\mathcal E_1 \cap \mathcal E_2$. The \emph{union} of subgraphs is defined similarly. Note that if $\mathcal G_{1,2}$ are induced, then their intersection, but not necessarily their union, is induced as well. In contrast, if $\mathcal G_{1,2}$ are connected and $\mathcal G_1 \cap \mathcal G_2 \ne \emptyset$, then their union, but not necessarily their intersection, is connected as well.

\subsection{Bipartite graphs}

Let us now specialise the above notions to bipartite graphs.
A graph $\mathcal G$ is \emph{bipartite} if the vertex-set is a disjoint union of two parts $\mathcal V = \widehat V \sqcup V$ and there are no edges within either part, i.e., $\mathcal E \cap \binom{\widehat V}2 = \emptyset$ and $\mathcal E \cap \binom V2 = \emptyset$. In other words, all the edges go between $\widehat V$ and $V$. We use $v,w,\ldots$ to denote the vertices in $V$ and $\hat v,\hat w,\ldots$ to denote vertices in $\widehat V$. Note that the distance $d(v,\hat v)$ is always odd, whereas $d(v,w)$ and $d(\hat v,\hat w)$ are always even.

The adjacency matrix of a bipartite graph takes the form
\ie\label{graph-adj-matrix}
\begin{pmatrix}
0_{|\widehat V|\times |\widehat V|} & \mathfrak h_{|\widehat V|\times |V|}\\
\mathfrak h_{|V|\times |\widehat V|}^\intercal & 0_{|V|\times |V|}
\end{pmatrix}~,
\fe
where $\mathfrak h$ is known as the \emph{biadjacency matrix}. It is given by $\mathfrak h_{\hat v,v} = 1$ if $v \sim \hat v$ and $0$ otherwise. Let $\mrdeg(\mathcal G) := \max_{v\in V} \deg(v)$ be the maximum among all the degrees of the vertices in $V$. We say that $\mathcal G$ is of \emph{bounded-right-degree} if there is a positive integer $D$ such that $\mrdeg(\mathcal G) \le D$. Similar comments apply to left-degrees, i.e., degrees of vertices in $\widehat V$.

For a bipartite graph $\mathcal G$, there are two kinds of automorphisms: (i) \emph{preserving} automorphisms which fix the vertex-sets $V$ and $\widehat V$, and (ii) \emph{reversing} automorphisms which exchange the vertex-sets $V$ and $\widehat V$.

A subgraph $\mathcal G'$ of the bipartite graph $\mathcal G$ is said to be \emph{left-closed} if for any $\hat v \in \widehat V'$, its neighbours in $\mathcal G'$ are the same as its neighbours in $\mathcal G$. A \emph{right-closed} subgraph is defined similarly. Note that the intersection and the union of two left-(right-)closed subgraphs are also left-(right-)closed.

\section{Review of classical LDPC code in transverse field}\label{sec:gTFIM}

In this section, we review local Hamiltonians based on classical LDPC codes in the presence of a transverse field. 

A classical LDPC code corresponds to a bipartite graph $\mathcal G$ on vertex-set $\widehat V \sqcup V$ with biadjacency matrix $\mathfrak h$.\footnote{In the quantum information literature, $\mathfrak h$ and $\mathcal G$ are called the \emph{parity check matrix} and the corresponding \emph{Tanner graph}, respectively.} The graph $\mathcal G$ is assumed to satisfy the following properties: (i) it is connected, and (ii) there are positive integers $D,\widehat D$ such that $\mldeg(\mathcal G) \le \widehat D$ and $\mrdeg(\mathcal G) \le D$.\footnote{The ``low-density'' in LDPC refers to the second assumption on $\mathcal G$, i.e., the left- and right-degrees of $\mathcal G$ are both bounded by constants. Equivalently, the parity check matrix $\mathfrak h$ is sparse, i.e., it has a bounded number of nonzero entries in each row and each column.} Together, these assumptions imply $|\widehat V| \le D |V|$ and $|V| \le \widehat D |\widehat V|$.

One can associate a quantum Hamiltonian to this code as follows. The total Hilbert space is the tensor product of qubits on the vertices in $V$, i.e., $\mathcal H := \bigotimes_{v\in V} \mathbb C^2$. The Pauli operators acting on the qubit $v$ are denoted as $Z_v$ and $X_v$. They satisfy $Z_v^2 = X_v^2 = 1$ and $Z_v X_v = - X_v Z_v$. The eigenstates of $Z_v$ are denoted as $|0\>$ and $|1\>$, whereas the eigenstates of $X_v$ are denoted as $|{+}\>$ and $|{-}\>$.

Each $\hat v \in \widehat V$ corresponds to a generalised Ising-type interaction term involving the qubits on the vertices adjacent to $\hat v$:\footnote{This is known as a \emph{(parity) check} in the quantum information literature.}
\ie
F_{\hat v} := \prod_{v\in V} Z_v^{\mathfrak h_{\hat v,v}} = \prod_{v\in V\,:\,v \sim \hat v} Z_v~.
\fe
The full Hamiltonian, including the transverse field terms, is given by
\ie\label{H-gTFIM}
H = - \sum_{\hat v\in\widehat V} J_{\hat v} F_{\hat v} - \sum_{v\in V} h_v X_v~,\qquad 
\fe
where $J_{\hat v}$ and $h_v$ are coupling constants. In \cite{Gorantla:2024ocs}, we referred to this model as the \emph{generalised transverse-field Ising model}.

The assumptions on the graph $\mathcal G$ have the following physical interpretation:
\begin{itemize}
\item Since $\mathcal G$ is connected, the model does not factorise into decoupled models.

\item Each interaction term involves at most $\widehat D$ qubits, and each qubit participates in at most $D$ interaction terms. One can interpret this as the graph-theoretic version of a ``local, short-ranged'' Hamiltonian.

\item We further assume that $|V|> 2D \widehat D$, which in turn implies $|\widehat V| > 2D$. Since we are interested in taking the thermodynamic (infinite-volume) limit, where $|V| \to \infty$ (and hence $|\widehat V|\to \infty$) while $D$ and $\widehat D$ are fixed, this is not an unreasonable assumption.
\end{itemize}

\subsection{Symmetries}\label{sec:gTFIM-sym}

The model \eqref{H-gTFIM} has several symmetries at various values of couplings. They can be organised into three categories.

\subsubsection{Internal symmetries (logical operators)}
There are several $\mathbb Z_2$ \emph{internal symmetries} generated by $\eta_a := \prod_{v\in V} X_v^{a_v}$,\footnote{In the quantum information literature, they are known as \emph{logical operators}.} where $a\in \ker \mathfrak h$, i.e., $\sum_{v\in V} \mathfrak h_{\hat v v} a_v = 0 \mod 2$. The condition on $a$ ensures that $\eta_a$ commutes with $F_{\hat v}$ for all $\hat v$, and hence it commutes with the Hamiltonian $H$ for any values of the couplings.

These operators are invertible and they satisfy the fusion rule
\ie
\eta_a \eta_{a'} = \eta_{a+a'}~,\qquad \forall a,a'\in \ker \mathfrak h~.
\fe
Therefore, the full internal symmetry group is $\mathbb Z_2^\nu$, where $\nu := \dim \ker \mathfrak h$.

\subsubsection{Spatial symmetries}
Every preserving automorphism $\pi$ of $\mathcal G$ corresponds to a \emph{spatial symmetry} generated by the operator $T_\pi$ that acts on local operators as
\ie
T_\pi X_v T_\pi^{-1} = X_{\pi(v)}~,\qquad T_\pi Z_v T_\pi^{-1} = Z_{\pi(v)}~.
\fe
In other words, the operator $T_\pi$ permutes the qubits according to the permutation $\pi|_V$, which explains the nomenclature ``spatial symmetries''. It acts on the generalising Ising interaction term as
\ie
T_\pi F_{\hat v} T_\pi^{-1} = F_{\pi(\hat v)}~.
\fe
Therefore, it commutes with $H$ provided $J_{\pi(\hat v)} = J_{\hat v}$ for all $\hat v\in \widehat V$ and $h_{\pi(v)} = h_v$ for all $v\in V$. One particular choice of couplings that satisfies these requirements is
\ie\label{J-h-choice}
J_{\hat v} = J |S(\hat v,r)|~,\quad h_v = h |S(v,r)|~,
\fe
for some $r\ge 0$ and constants $J,h$. Specifically, setting $r=0$ gives the more familiar model where $J_{\hat v} = J$ and $h_v = h$.

These operators are also invertible and they satisfy the fusion rule
\ie
T_\pi T_{\pi'} = T_{\pi \circ \pi'}~,
\fe
where $\pi,\pi'$ are preserving automorphisms.

\subsubsection{Non-invertible duality symmetries}
Every reversing automorphism $\rho$ of $\mathcal G$ corresponds to a \emph{non-invertible duality symmetry} generated by the operator $\mathsf D_\rho$. More concretely, it is a composition of the reversing automorphism with the gauging map that gauges the internal $\mathbb Z_2^\nu$ symmetry. (See \cite{Gorantla:2024ocs} for a discussion of this operator, its algebra, and a tensor network representation based on the graph $\mathcal G$.) It acts on the local $\mathbb Z_2^\nu$-symmetric operators as
\ie
\mathsf D_\rho X_v = F_{\rho(v)} \mathsf D_\rho~,\qquad \mathsf D_\rho F_{\hat v} = X_{\rho(\hat v)} \mathsf D_\rho~.
\fe
It is referred to as a ``duality symmetry'' because it is built from the $\mathbb Z_2^\nu$-gauging-map that gauges the $\mathbb Z_2^\nu$ internal symmetries and exchanges the generalised Ising and transverse-field terms in the Hamiltonian. It commutes with $H$ provided $J_{\rho(v)} = h_v$ for all $v\in V$ and $h_{\rho(\hat v)} = J_{\hat v}$ for all $\hat v\in \widehat V$. One particular choice of couplings that satisfies these requirements is the one in \eqref{J-h-choice} with $J = h$.

These operators are non-invertible and they satisfy the fusion rule
\ie\label{gTFIM-D^2}
\mathsf D_\rho \mathsf D_{\rho'} = T_{\rho \circ \rho'} \mathsf C~,\qquad \mathsf C \sim \sum_{a\in \ker \mathfrak h} \eta_a~,
\fe
where $\rho,\rho'$ are reversing automorphisms, and we omitted a constant factor in the definition of the condensation operator $\mathsf C$.

For completeness, we note the rest of the fusion algebra here:
\ie\label{gTFIM-fusion-alg}
&T_\pi \eta_a = \eta_{\pi(a)} T_\pi~,\qquad \mathsf D_\rho \eta_a = \eta_a \mathsf D_\rho = \mathsf D_\rho~,
\\
&T_\pi \mathsf D_\rho = \mathsf D_{\pi\circ\rho}~,\qquad \mathsf D_\rho T_\pi = \mathsf D_{\rho\circ\pi}~,
\fe
where $\pi(a)_v := a_{\pi^{-1}(v)}$.

\subsection{Ground states}
Solving for the spectrum of the Hamiltonian $H$ at generic couplings is hard. However, the model is exactly solvable in some special cases. For instance, when $J_{\hat v} = 0$ and $h_v>0$, $H$ has only the transverse-field terms, so there is a unique ground state,
\ie\label{plus-state}
|\Plus\> := |{+}\cdots{+}\>~.
\fe
This corresponds to the $\mathbb Z_2^\nu$-preserving (paramagnetic) phase.

On the other hand, when $J_{\hat v}>0$ and $h_v = 0$, $H$ has only the generalised Ising interaction terms, and in fact, it is the commuting projector Hamiltonian based on the classical LDPC code. There are $2^\nu$ ground states, one for each $a\in \ker\mathfrak h$, that form the code space:
\ie\label{graph-eta-state}
|a\>:= \eta_a |\Zero\>~, \quad\text{where}\quad |\Zero\> := |0\cdots 0\>,
\fe
This corresponds to the $\mathbb Z_2^\nu$-broken (ferromagnetic) phase. The two phases are exchanged by the non-invertible duality symmetries, i.e.,\footnote{The action of $\mathsf D_\rho$ on $|\Zero\>$ can be computed using ZX-calculus:
\ie
\mathsf D_\rho |\Zero\> \sim \begin{ZX}[circuit]
\zxN{} &[-8pt] \zxN-{v\in V} & \zxN-{\hat v\in \widehat V} & \zxN{} &\zxN-{v\in V} &[-8pt]\zxN{}\\[-8pt]
\zxX{} \rar & \zxZ{} \ar[d,3 vdots] \ar[dr] & \zxX{} \ar[dl,"\mathfrak h" {description,font=\normalsize}] \ar[d,3 vdots] \rar[H] & \zxN{} \ar[dr] & \zxN{} \ar[d,3 vdots] \ar[dl,"\rho" {description,font=\normalsize}] \rar & \zxN{}\\[8pt]
\zxX{} \rar & \zxZ{} & \zxX{} \rar[H] & \zxN{} &\zxN{} \rar & \zxN{}\\
\end{ZX}
\overset{\text{SC,SF}}{\sim} \begin{ZX}[circuit]
\zxN-{\hat v\in \widehat V} & \zxN{} &\zxN-{v\in V} &[-8pt]\zxN{}\\[-8pt]
\zxX{} \ar[d,3 vdots] \rar[H] & \zxN{} \ar[dr] & \zxN{} \ar[d,3 vdots] \ar[dl,"\rho" {description,font=\normalsize}] \rar & \zxN{}\\[8pt]
\zxX{} \rar[H] & \zxN{} &\zxN{} \rar & \zxN{}\\
\end{ZX}
\overset{\text{CC'}}{=} \begin{ZX}[circuit]
\zxN-{v\in V} &[-8pt]\zxN{}\\[-8pt]
\zxZ{} \ar[d,3 vdots] \rar & \zxN{}\\[8pt]
\zxZ{} \rar & \zxN{}\\
\end{ZX}
= |\Plus\>~.
\fe
Since $\mathsf D_\rho \eta_a = \mathsf D_\rho$ \eqref{gTFIM-fusion-alg}, the action on $|a\>$ is similar. The action of $\mathsf D_\rho$ on $|\Plus\>$ can then be inferred using the fact that $\mathsf D_\rho \mathsf D_{\rho^{-1}} = \mathsf C$ \eqref{gTFIM-D^2}, and $\mathsf C|\Zero\> \sim \sum_{a\in \mathfrak h} |a\>$. For a quick review of ZX-calculus, see \cite{vandeWetering:2020giq}, and for an application of ZX-calculus to non-invertible symmetries---in particular, for examples of derivations similar to the ones here---see \cite[Section 6]{Gorantla:2024ocs}.}
\ie\label{gTFIM-D-action}
\mathsf D_\rho |\Plus\> \sim \sum_{a\in \ker \mathfrak h} |a\>~,\qquad \mathsf D_\rho |a\> \sim |\Plus\>~,
\fe
where we omitted some constant factors.

One can write the ground states \eqref{graph-eta-state} as
\ie\label{graph-xistate}
|\xi\> := \prod_{a\in \mathcal B} \eta_a^{\xi_a} |\Zero\>~,
\fe
where $\xi \in \{0,1\}^\nu$ and $\mathcal B$ is a basis of $\ker \mathfrak h$ over $\mathbb Z_2$. It is also useful to write these ground states in a different basis:
\ie
|\zeta\> := \frac1{2^{\nu/2}} \sum_{\xi \in \{0,1\}^\nu} (-1)^{\sum_{a\in \mathcal B} \zeta_a \xi_a} |\xi\>
= 2^{\nu/2} \prod_{a\in \mathcal B} \left(\frac{1+(-1)^{\zeta_a} \eta_a}2\right) |\Zero\>~,
\fe
where $\zeta \in \{0,1\}^\nu$. The $\mathbb Z_2^\nu$ symmetry operators are diagonal in this basis:
\ie
\prod_{a\in \mathcal B} \eta_a^{\xi_a}~|\zeta\> = (-1)^{\sum_{a\in\mathcal B}\xi_a\zeta_a} |\zeta\>~.
\fe
That is, the state $|\zeta\>$ carries the charge $\zeta_a$ under $\eta_a$ for each $a\in \mathcal B$.

Next, we want to write the states $|\zeta\>$ in terms of the charged operators. Define the \emph{charged operator} as
\ie
W_b := \prod_{v\in V} Z_v^{b_v}~,
\fe
where $b\in \mathbb Z_2^{|V|}$. It carries the charge $b\cdot a = \sum_{v\in V} b_v a_v$ under the internal symmetry generated by $\eta_a$ because
\ie
\eta_a W_b \eta_a^{-1} = (-1)^{b\cdot a} W_b~.
\fe
Note that, since $F_{\hat v}$ commutes with $\eta_a$, $W_b F_{\hat v}$ and $W_b$ carry the same charges. In fact, the converse is also true, i.e., if $W_b$ and $W_{b'}$ carry same charges, then $W_{b'}W_b^{-1}$ is a product of some $F_{\hat v}$'s.\footnote{The distinct charged operators are labelled by elements of the quotient space $(\ker \mathfrak h)^* := \mathbb Z_2^{|V|}/(\ker \mathfrak h)^\perp$, known as the \emph{dual space} of $\ker \mathfrak h$. Here, $(\ker \mathfrak h)^\perp$ is the \emph{orthogonal complement} of $\ker \mathfrak h$, i.e., the set of all $b\in \mathbb Z_2^{|V|}$ such that $b\cdot a = 0$ for all $a\in \ker \mathfrak h$. The orthogonal complement is generated by the generalised Ising interaction terms $F_{\hat v}$.}

Let $\mathcal B^*$ be a dual (or reciprocal) basis of $(\ker \mathfrak h)^*$. One can pick a dual basis so that for any $a\in \mathcal B$ there is a unique $a^* \in \mathcal B^*$ that satisfies $a^* \cdot b = \delta_{a,b}$ for all $b \in \mathcal B$, where $\delta_{a,b} = 1$ if $a=b$ and $0$ if $a\ne b$.

In terms of the charged operators associated with the dual basis, we can write\footnote{We can derive \eqref{zeta-basis-ch-op} using ZX-calculus (see \cite[Section 6]{Gorantla:2024ocs} for similar manipulations):
\ie\label{zeta-basis-zx-deriv}
|\zeta=0\> \sim \prod_{a\in \mathcal B} \left(\frac{1+\eta_a}2\right) |\Zero\> \sim \begin{ZX}[circuit,mbr=4.5]
\zxN{} & \zxN-{v\in V} & \zxN-{\hat v \in \widehat V}
\\[-8pt]
\zxN{} & \zxZ{} \ar[d,3 vdots] \ar[dr] & \zxX{} \ar[d,3 vdots] \ar[dl,"\mathfrak h"{description,font=\normalsize}]
\\[8pt]
\zxN{} & \zxZ{} & \zxX{}
\\
\zxX{} \rar & \zxX{} \ar[d,3 vdots] \ar[uu,('] \ar[r] & \zxN{}
\\[8pt]
\zxX{} \rar & \zxX{} \ar[uu,('] \ar[r] & \zxN{}
\\[-8pt]
\zxN{} & \zxN-{v\in V} & \zxN{}
\end{ZX}
\overset{\text{SF,I}}{=} \begin{ZX}[circuit,mbr=4.5]
\zxN{} & \zxN-{v\in V} & \zxN-{\hat v \in \widehat V}
\\[-8pt]
\zxN{} & \zxZ{} \ar[d,3 vdots] \ar[dr] & \zxX{} \ar[d,3 vdots] \ar[dl,"\mathfrak h"{description,font=\normalsize}]
\\[8pt]
\zxN{} & \zxZ{} & \zxX{}
\\
\zxZ{} \rar & \zxZ{} \ar[d,3 vdots] \ar[uu,('] \ar[r] & \zxN{}
\\[8pt]
\zxZ{} \rar & \zxZ{} \ar[uu,('] \ar[r] & \zxN{}
\\[-8pt]
\zxN{} & \zxN-{v\in V} & \zxN{}
\end{ZX}
\sim \prod_{\hat v \in \widehat V} \left(\frac{1+F_{\hat v}}2\right) |\Plus\>~.
\fe}
\ie\label{zeta-basis-ch-op}
|\zeta\> \sim \prod_{\hat v \in \widehat V} \left(\frac{1+F_{\hat v}}2\right) \prod_{a\in \mathcal B} W_{a^*}^{\zeta_a}|\Plus\>~,
\fe
where we omit the normalisation factor. This follows from \eqref{zeta-basis-zx-deriv}. Expanding the product over $\hat v$, we can rewrite this state as
\ie\label{zeta-eq-wt-sup}
|\zeta\> \sim \sum_{\widehat U \subseteq \widehat V}~ \prod_{\hat v\in \widehat U} F_{\hat v} ~ \prod_{a\in \mathcal B} W_{a^*}^{\zeta_a}|\Plus\> = \sum_{\widehat U \subseteq \widehat V} |\sigma^{\zeta,\widehat U}\>~,
\fe
where for each subset $\widehat U \subseteq \widehat V$ and each $\zeta\in\{0,1\}^\nu$, we define the state
\ie\label{eta-product-eigenstates}
|\sigma^{\zeta,\widehat U}\> := \prod_{\hat v\in \widehat U} F_{\hat v} |\sigma^\zeta\>~,\qquad \text{where}\qquad |\sigma^\zeta\> := \prod_{a\in \mathcal B} W_{a^*}^{\zeta_a}|\Plus\>~.
\fe

The utility of these states comes from the following interpretation. First, $|\sigma^\zeta\>$ is a product state with a $|{+}\>$ or $|{-}\>$ at each vertex, so it represents a configuration of ``signs'' on the vertices $V$. In other words, $|\sigma^\zeta\>$ is an eigenstate of $X_v$'s, so we refer to it as an \emph{X-state}. Acting by $F_{\hat v}$ on this state ``flips the signs'' on all vertices adjacent to $\hat v$. Then, the state $|\sigma^{\zeta,\widehat U}\>$ is obtained by applying such flips on all $\hat v \in \widehat U$. Since $F_{\hat v}$ commutes with $\eta_a$, the flips do not change the charges carried by the state under $\mathbb Z_2^\nu$, i.e., $|\sigma^{\zeta,\widehat U}\>$ carries the charge $\zeta_a$ under $\eta_a$ for all $a\in \mathcal B$, independent of $\widehat U$.

In fact, the converse is also true, i.e., given any two X-states $|\sigma\>$ and $|\sigma'\>$ that carry the same charges, there is a ``set of flips'' $\widehat U \subseteq \widehat V$ such that $|\sigma'\> = \prod_{\hat v\in \widehat U} F_{\hat v} |\sigma\>$. To see this, write $|\sigma\> = W_b |\Plus\>$ and $|\sigma'\> = W_{b'} |\Plus\>$ for some $b,b'\in \mathbb Z_2^{|V|}$. Then, $W_b$ and $W_{b'}$ carry the same charges, so $W_{b'} W_b^{-1}$ is uncharged. This means it can be written as a product of $F_{\hat v}$'s, i.e., there is a $\widehat U \subseteq \widehat V$ such that $W_{b'} W_b^{-1} = \prod_{\hat v \in \widehat U} F_{\hat v}$.

Finally, the expression for the state $|\zeta\>$ in \eqref{zeta-eq-wt-sup} can be interpreted as an equal-weight superposition of all X-states $|\sigma\>$ that carry the same charges under $\mathbb Z_2^\nu$ as $|\zeta\>$.

\section{Deformed classical LDPC code in transverse field}\label{sec:graph-deformed}
We now propose a deformation of the classical LDPC code in transverse field that preserves all its symmetries. Its ground states include a trivial product state and the code space. Moreover, they are exactly degenerate at finite volume and remain gapped in the thermodynamic (infinite-volume) limit. In other words, the deformed model realises a coexistence of the trivial phase and the $\mathbb Z_2^\nu$-symmetry-broken phase. In fact, it realises a gapped phase where the non-invertible duality symmetries are spontaneously broken.

The deformation we propose is the following:\footnote{The choice of distance $3$ is not special. One could choose a different distance that is odd and greater than $1$, and our conclusions remain unchanged.}
\ie\label{graph-deformedH}
H_\lambda = \sum_{v,\hat v\,:\, d(v,\hat v)=3} \left( - J F_{\hat v} - h X_v + \frac{J\lambda}2 X_v F_{\hat v} \right)~,
\fe
where $\lambda$ is the coupling associated with the deformation. The first thing to note about this Hamiltonian is that the deformation is short-ranged and local, i.e., each term in the Hamiltonian involves only a finite number of qubits and each qubit participates in a finite number of terms, both independent of the system size.

Next, the Hamiltonian $H_{\lambda = 0}$ is just a special case of $H$ \eqref{H-gTFIM} with a particular choice of couplings $J_{\hat v}$ and $h_v$---specifically, the coupling constants are given by setting $r=3$ in \eqref{J-h-choice}. Therefore, $H_\lambda$ is indeed a deformation of $H$ that preserves all of its symmetries.

Finally, since each term of the Hamiltonian is a product of local $\mathbb Z_2^\nu$-symmetric operators, it is manifestly invariant under the $\mathbb Z_2^\nu$ symmetry. Moreover, it is invariant under the spatial symmetries, as well as the non-invertible duality symmetries if we further set $J=h$, because for any (preserving or reversing) automorphism $\phi$ of $\mathcal G$, we have $d(\phi(v),\phi(\hat v)) = d(v,\hat v)$.

In the rest of this section, we focus on the point $J=h$ and $\lambda = 1$. Then, the Hamiltonian can be written as
\ie\label{graph-deformedH'}
H' = 2J \sum_{v,\hat v\,:\, d(v,\hat v)=3} P_v Q_{\hat v}~,
\fe
where we defined the orthogonal projections $P_v := \frac12(1-X_v)$ and $Q_{\hat v} := \frac12(1-F_{\hat v})$.

\subsection{Exact ground states and frustration-freeness}\label{sec:exact-gnd-st}
Note that $P_v$ commutes with $Q_{\hat v}$ if and only if $v\not\sim \hat v$. Since product of commuting orthogonal projections is positive semi-definite, each term in $H'$ is positive semi-definite, and hence $H'$ is itself positive semi-definite, i.e., $H'\ge 0$. Moreover, it is easy to see that
\ie
H' |\Plus\> = 0~,\qquad H' |a\> = 0~,
\fe
where $|\Plus\>$ and $|a\>$ is defined in \eqref{plus-state} and \eqref{graph-eta-state}, respectively. Therefore, these states are ground states of $H'$; in particular, $H'$ is frustration-free.

In the following sections, we show that there are no other ground states and that these ground states remain gapped in the infinite volume limit, i.e., $|V| \to \infty$ (which implies $|\widehat V| \to \infty$) with $D,\widehat D$ held fixed.

If the original Hamiltonian $H$ has a duality symmetry, then $H'$ also preserves it. In that case, by the action of the non-invertible duality symmetry operators on the above ground states \eqref{gTFIM-D-action}, we conclude that $H'$ provides a concrete lattice realisation of a gapped phase where the duality symmetry is spontaneously broken.

\subsection{Proof of ground state degeneracy}\label{sec:graph-proofofGSD}
Since the Hamiltonian $H'$ is frustration-free, any ground state $|\psi\>$ is annihilated by $H'$. Since $H'$ is a sum of positive semi-definite terms, any state in the kernel of $H'$ is annihilated by each term separately, i.e., for any $v,\hat v$ such that $d(v,\hat v)=3$, we have
\ie\label{graph-local-constraint}
P_v Q_{\hat v} |\psi\> = 0~.
\fe

Let us decompose $|\psi\>$ in the eigenbasis of $X_v$'s:
\ie\label{psi-decomp}
|\psi\> = \sum_{\sigma \in \{+,-\}^{|V|}} \psi_\sigma |\sigma\>~,
\fe
where, as mentioned below \eqref{eta-product-eigenstates}, $|\sigma\>$ is an X-state associated with a configuration of ``signs'' on the vertices $V$, and $\psi_\sigma \in \mathbb C$ is the ``wave function'' in this sign-configuration space. Consider the action of $P_v Q_{\hat v}$ on the X-state $|\sigma\>$:
\ie
&P_v Q_{\hat v} |\sigma_v\sigma_{v_1}\cdots \sigma_{v_k}\> = \begin{cases}0~, & \sigma_v = +~,
\\
\frac12 \left( |\sigma_v\sigma_{v_1}\cdots \sigma_{v_k}\> - |\sigma_v\bar\sigma_{v_1}\cdots \bar\sigma_{v_k}\> \right)~, & \sigma_v = -~.
\end{cases}
\fe
Here, we show only the vertices $v,v_1,\ldots,v_k$ explicitly---where $v_1,\ldots,v_k$ are the vertices adjacent to $\hat v$---and omit the other vertices. We also use $\sigma_w = \pm$ to denote the sign on the vertex $w\in V$ and write $\bar\sigma_w := -\sigma_w$. Then, the constraint \eqref{graph-local-constraint} implies that
\ie\label{psi-constraint}
\psi_{\sigma_v\sigma_{v_1}\cdots \sigma_{v_k}} = \psi_{\sigma_v\bar\sigma_{v_1}\cdots \bar\sigma_{v_k}}~,\quad \text{if} \quad \sigma_v = -~,
\fe
where, once again, we omit the other vertices. In particular, $\psi_{\Plus}$ is unconstrained because $|\Plus\>$ does not have a vertex with $-$.

Let us rephrase the above constraints. Given an X-state $|\sigma\> \ne |\Plus\>$, a vertex $v$ with $\sigma_v = -$, and a vertex $\hat v$ such that $d(v,\hat v)=3$, let $|\sigma'\> = F_{\hat v} |\sigma\>$ be the X-state obtained by flipping the signs on all the vertices adjacent to $\hat v$. Then, the constraint \eqref{graph-local-constraint} implies that $\psi_{\sigma'} = \psi_\sigma$.
In fact, we claim that
\begin{claim}\label{clm:eq-wt}
If two X-states $|\sigma\>, |\sigma'\> \ne |\Plus\>$ carry the same charges under $\mathbb Z_2^\nu$, then $\psi_{\sigma'} = \psi_\sigma$.
\end{claim}
\noindent We defer the proof of this claim to Appendix \ref{app:eq-wt}. What the claim means is that all X-states, except $|\Plus\>$, with equal charges under $\mathbb Z_2^\nu$ have equal weights in the decomposition \eqref{psi-decomp}. Recall that the equal-weight superposition of all X-states with equal charges under $\mathbb Z_2^\nu$ is precisely the state $|\zeta\>$ in \eqref{zeta-eq-wt-sup}. Therefore, we can write
\ie
|\psi\> &= \alpha_{\Plus} |\Plus\> + \sum_{\zeta\in \{0,1\}^\nu} \alpha_\zeta |\zeta\>~,
\fe
for some $\alpha_{\Plus},\alpha_\zeta \in \mathbb C$.

\subsection{Proof of gap}

In this section, we prove that $H'$ \eqref{graph-deformedH'} is gapped in the thermodynamic (infinite-volume) limit, i.e., $|V| \to \infty$ (and hence, $|\widehat V| \to \infty$) while $D,\widehat D$ are held fixed.

\subsubsection{Overview of the martingale method}\label{sec:mart-meth}
In general, proving that a local Hamiltonian on a tensor product Hilbert space is gapped in the thermodynamic limit is a very hard problem, and in fact, undecidable \cite{Cubitt:2015lta}. Fortunately, when the Hamiltonian is frustration-free, there are a few methods to prove the gap . The one we use here is the \emph{martingale method} \cite{Fannes:1992}. The philosophy behind this method is summarised in the following steps:
\begin{itemize}
\item \emph{Knabe's argument \cite{Knabe:1988}:} Say $H$ is a positive semi-definite operator with a zero eigenvalue and let $\gamma(H)>0$ be its smallest nonzero eigenvalue, a.k.a., the gap of $H$. Then, for any $\epsilon>0$, $H^2 \ge \epsilon H \implies \gamma(H) \ge \epsilon$. Therefore, in order to prove that $H$ is gapped in the thermodynamic limit, it suffices to find an $\epsilon>0$ independent of the system size such that $H^2 \ge \epsilon H$ for all sufficiently large systems.

\item \emph{Coarse-graining (blocking):} Let us schematically write the Hamiltonian as $H = \sum_i H_i$, where $i$ labels the local interaction terms. We can divide the full system into many finite-sised \emph{blocks}, each of which contains several interaction terms. The size of each block can be as large as needed for the problem at hand, but should be {\color{blue}finite and independent of the system size}. Moreover, the blocks need not be disjoint, so an interaction term may be contained in multiple blocks. If each interaction is contained in {\color{blue}at most $m$ blocks}, then the Hamiltonian satisfies $H \ge \frac1m \sum_I H_I$, where $I$ labels the blocks and $H_I = \sum_{i\in I} H_i$. So, proving the gap for $\sum_I H_I$ proves the gap for $H$. Note that since $H$ is positive semi-definite and frustration-free, so is each $H_I$, and also $\sum_I H_I$.

\item \emph{Reduce to projections:} Let $\Pi_I$ be the orthogonal projection onto the ground state space of $H_I$, so that $H_I \ge \gamma(H_I) \Pi_I^\perp$. Now, {\color{blue}if there is an $\alpha>0$, independent of the system size, such that $\gamma(H_I) \ge \alpha$ for all $I$},\footnote{This follows from the assumptions that (i) the maximum of the block sizes is independent of the system size and (ii) the block Hilbert space $\mathcal H_I = \bigotimes_{v\in I} \mathcal H_v$ is finite dimensional.} then $\sum_I H_I \ge \alpha \sum_I \Pi_I^\perp$. So it suffices to prove the gap for $\bar H := \sum_I \Pi_I^\perp$. By Knabe's argument, we have to show that there is an $\epsilon >0$ independent of the system size such that $\bar H^2 \ge \epsilon \bar H$ for all sufficiently large systems. We have,
\ie
\bar H^2 &= \sum_I (\Pi_I^\perp)^2 + \frac12\sum_{I\ne J}' \{\Pi_I^\perp,\Pi_J^\perp\} + \frac12\sum_{I\ne J}'' \{\Pi_I^\perp,\Pi_J^\perp\}
\\
&\ge \bar H + \frac12\sum_{I\ne J}' \{\Pi_I^\perp,\Pi_J^\perp\}~,
\fe
where the sum with a single prime involves projections $\Pi_I^\perp$ and $\Pi_J^\perp$ that do not commute, whereas the sum with two primes involves those that commute. The inequality follows from the fact that product of commuting projections is positive semi-definite.

\item \emph{Martingale condition:} There are two reasons for reducing to projections in the last step: (i) the first term in $\bar H^2$ reduces to $\bar H$ itself, and (ii) we can exploit the following identities satisfied by any two projections $\Pi_{1,2}$:
\ie
&\{ \Pi_1, \Pi_2\} \ge - \Vert  \Pi_1  \Pi_2 -  \Pi_1 \wedge  \Pi_2 \Vert~( \Pi_1 +  \Pi_2)~,
\\
&\Vert  \Pi_1  \Pi_2 -  \Pi_1 \wedge  \Pi_2 \Vert = \Vert  \Pi_1^\perp  \Pi_2^\perp -  \Pi_1^\perp \wedge  \Pi_2^\perp \Vert~,
\fe
where $\Pi_1 \wedge \Pi_2$ is the projection onto $\im(\Pi_1) \cap \im(\Pi_2)$. Using these identities for the sum with single prime, we get
\ie
\bar H^2 \ge \sum_I \left(1-\sum_{J\ne I}'\delta_{I,J}\right) \Pi_I^\perp~,
\fe
where we defined
\ie
\delta_{I,J} := \Vert  \Pi_I  \Pi_J -  \Pi_I \wedge  \Pi_J \Vert~,
\fe
and refer to it as the \emph{martingale function}. Now, {\color{blue}if there is a $0\le\beta<1$, independent of the system size, such that $\sum_{J\ne I}' \delta_{I,J} \le \beta$ for all $I$}, which we refer to as the \emph{martingale condition}, then $\bar H^2 \ge (1-\beta) \bar H$, thereby proving the gap.
\end{itemize}
Some comments are in order:
\begin{enumerate}
\item One might wonder what the purpose of coarse-graining was. This can be appreciated in the last step. In any physical system with local interactions, it is typically the case that $\delta_{I,J}$ gets smaller as the size of the intersection $I\cap J$ increases, i.e., when $I$ and $J$ share more and more qubits. Indeed, for any gapped system on a Euclidean lattice, it is known that $\delta_{I,J} \sim e^{-c |I\cap J|^r}$ where $c,r$ are some positive constants \cite{Kastoryano_2018}. Conversely, if $\delta_{I,J}$ has this behaviour, then one can satisfy the martingale condition $\sum_{J\ne I}' \delta_{I,J} \le \beta < 1$ for all $I$ by ensuring that each block intersects only a finite number (independent of the system size) of other blocks, and each such intersection is large. A larger intersection means larger blocks, and hence the need for coarse-graining.

\item While making the intersections arbitrarily large might seem like a good idea, the flip side is that larger blocks lead to smaller $\alpha$. So we cannot make the blocks arbitrarily large. In particular, the maximum size of a block must be independent of the system size---else, it is possible that $\alpha \to 0$ in the thermodynamic limit making the inequality $\sum_I H_I \ge \alpha \sum_I \Pi_I^\perp$ useless.

\item In a similar vein, the blocks cannot be too dense/crowded, i.e., too many blocks cannot overlap simultaneously, because that would increase the value of $m$ making the inequality $H \ge \frac1m \sum_I H_I$ useless.\footnote{This is usually not an issue on Euclidean lattices where each block intersects a finite number (that depends mostly only on the dimension and coordination number of the lattice) of adjacent blocks. However, on hyperbolic lattices, this can be a major issue.} In particular, we need $m$ to be independent of the system size as well. But if the blocks are too far apart, then the intersection between them is too small, making it harder to satisfy the martingale condition.
\end{enumerate}
One can see that there is a delicate balance between various aspects of the coarse-graining procedure, such as the size of a block, the number of blocks adjacent to it, and the size of intersection of two adjacent blocks. This balance can be struck on Euclidean lattices; see \cite{Nachtergaele:1996,Hastings:2005pr,Bravyi:2015ksj,Gosset:2016,Kastoryano_2018,Lemm:2019pxa,Abdul-Rahman:2019wna,Anshu:2019mhi,Lemm:2019jxn,Pomata:2019uap,Andrei:2022yym,Young_2024,Lemm:2024yre,Rai:2024rge} for some examples in diverse dimensions.

However, our model is defined on an arbitrary bipartite graph (with fixed $D$ and $\widehat D$). While the martingale method can be extended to this scenario in principle, to the best of our knowledge, it has never actually been attempted before.\footnote{Commuting projector models are an exception, but they are obviously gapped. Our statement is about frustration-free Hamiltonians that are not commuting projector models.}

\subsubsection{Coarse-graining a bipartite graph}\label{sec:coarse-grain-bi-graph}
Our first challenge then is to coarse-grain a bipartite graph $\mathcal G$. As explained above, there is a delicate balance that needs to be struck, and \emph{a priori}, it is not clear that this should be possible for arbitrary bipartite graphs (with fixed $D$ and $\widehat D$). One of our main technical achievements is to show that it is, in fact, possible. We defer the details to Appendix \ref{app:coarse-grain} and summarise the upshot below. We note that this summary is tailored for the model at hand, whereas the appendix considers a more general setting.

Given a positive interger $n> 2D \widehat D$ and a bipartite graph $\mathcal G$ with $|V|\ge n$, $\mldeg(\mathcal G) \le \widehat D$, and $\mrdeg(\mathcal G) \le D$, there is a collection $\{ \mathcal G_i \}$ of subgraphs of $\mathcal G$ that satisfy the following conditions:
\begin{enumerate}
\item each $\mathcal G_i$ is connected, induced, and left-closed,

\item for any $v\in V$ and $\hat v \in \widehat V$ such that $d(v,\hat v)=3$, every shortest path between $v$ and $\hat v$ is contained in at least one $\mathcal G_i$ and at most $m = m_{D,\widehat D}(n)$ distinct $\mathcal G_i$'s,

\item for each $i$, $n\le |V_i| \le N = N_{D,\widehat D}(n)$, and

\item for each $i\ne j$, if $\mathcal G_i$ and $\mathcal G_j$ intersect nontrivially, then $|V_i \cap V_j| \ge n$ and $\widehat V_i \cap \widehat V_j \ne \emptyset$.
\end{enumerate}
We refer to such a collection as a \emph{good cover} of $\mathcal G$ with parameters $m_{D,\widehat D}(n)$ and $N_{D,\widehat D}(n)$. The explicit expressions for the parameters can be inferred from the discussion in Appendix \ref{app:coarse-grain}---what is important is that they are independent of the system size $|V|$ and grow only polynomially in $n$ for fixed $D$ and $\widehat D$.

The independent parameter $n$ controls the size of a block, the number of blocks adjacent to it, and the size of intersection of two adjacent blocks. Making it larger, but still keeping it independent of the system size, is how we achieve the balance needed to satisfy the martingale condition.

\subsubsection{Proof of gap using the martingale method}\label{sec:proofofgap-deformedmodel}
We can now apply the above steps of the martingale method to prove the gap of $H'$ \eqref{graph-deformedH'}. Define the \emph{local Hamiltonian} on $\mathcal G_i$ as
\ie
H'_i = 2J\sum_{v,\hat v\in \mathcal G_i\,:\,d(v,\hat v) = 3} P_v Q_{\hat v}~,
\fe
Since $\mathcal G_i$ is connected and $|V_i|\ge n \ge 2D \widehat D$ by properties 1 and 3 of the good cover, $H'_i$ has precisely $2^{\nu_i}+1$ zero-energy states, where $\nu_i := \dim \ker \mathfrak h_i$ and $\mathfrak h_i$ is the biadjacency matrix of $\mathcal G_i$.\footnote{The ground states of $H'_i$ can be computed explicitly as in Section \ref{sec:graph-proofofGSD}.} Due to properties 1 and 2 of the good cover, we have
\ie
H' \ge \frac1{m_{D,\widehat D}(n)} \sum_i H'_i~.
\fe
Let $\Pi_i$ be the projection onto the nontrivial kernel of $H'_i$ and let $\Pi_i^\perp := 1-\Pi_i$. Since $|V_i|\le N_{D,\widehat D}(n)$ by property 3 of the good cover, there is an $\alpha_{D,\widehat D}(n) > 0$ such that $\gamma(H'_i) \ge \alpha_{D,\widehat D}(n)$ for all $i$. Therefore, we have $H'_i \ge \alpha_{D,\widehat D}(n) \Pi_i^\perp$ and hence, 
\ie
H' \ge \frac{\alpha_{D,\widehat D}(n)}{m_{D,\widehat D}(n)} \bar H~,\qquad \text{where} \qquad \bar H := \sum_i \Pi_i^\perp~.
\fe
It suffices to show that $\bar H$ is gapped in the infinite volume limit.

We now use Knabe's argument. Consider the square of the Hamiltonian $\bar H$,
\ie
\bar H^2 = \sum_i \Pi_i^\perp + \frac12 \sum_i \sum_{j\ne i\,:\,V_i \cap V_j \ne \emptyset} \{\Pi_i^\perp, \Pi_j^\perp\} + \frac12 \sum_i \sum_{j\ne i\,:\,V_i \cap V_j = \emptyset} \{\Pi_i^\perp, \Pi_j^\perp\}~.
\fe
The third term on the right hand side is positive semi-definite, but the second term may not be because $\Pi_i^\perp$ commutes with $\Pi_j^\perp$ if and only if $V_i \cap V_j = \emptyset$. Instead, for the second term, we have
\ie
\{ \Pi_i^\perp, \Pi_j^\perp\} \ge - \delta_{i,j} ( \Pi_i^\perp +  \Pi_j^\perp)~,
\fe
where we defined $\delta_{i,j} := \Vert  \Pi_i  \Pi_j -  \Pi_i \wedge  \Pi_j \Vert$. It follows that
\ie
\frac12 \sum_i \sum_{j\ne i\,:\,V_i \cap V_j \ne \emptyset} \{\Pi_i^\perp, \Pi_j^\perp\} \ge -\sum_i \Pi_i^\perp \sum_{j\ne i\,:\,V_i \cap V_j \ne \emptyset} \delta_{i,j}~,
\fe

In Appendix \ref{sec:graph-delta}, we derive an upper bound on $\delta_{i,j}$ in terms of $
|V_i\cap V_j|$,
\ie\label{up-bnd-delta}
\delta_{i,j} \le c \times 2^{-|V_i \cap V_j|/4\max(D,\widehat D)}~,
\fe
where $c$ is a positive constant, whose precise value can be found in \eqref{graph-delta-upper-bound}. This bound holds whenever $V_i \cap V_j$ and $\widehat V_i \cap \widehat V_j$ are both non-empty, which is indeed the case by property 4 of the good cover. Using properties 2, 3, and 4 of the good cover along with this inequality, we have
\ie
\sum_{j\ne i\,:\,V_i \cap V_j \ne \emptyset} \delta_{i,j} &\le c \sum_{j\ne i\,:\,V_i \cap V_j \ne \emptyset} 2^{-|V_i \cap V_j|/4\max(D,\widehat D)}
\\
&\le c \times N_{D,\widehat D}(n) \times D^2\widehat D^2(1+\widehat D) \times m_{D,\widehat D}(n) \times 2^{-n/4\max(D,\widehat D)}
\\
&=: \beta_{D,\widehat D}(n)~.
\fe
The prefactor in the second line comes from the following counting argument: each $\mathcal G_i$ contains at most $N_{D,\widehat D}(n)$ vertices of $V$, each of which participates in at most $D^2\widehat D^2(1+\widehat D)$ interaction terms (see Example 2 in Appendix \ref{app:int-hypergraph} for this counting), each of which is contained in at most $m_{D,\widehat D}(n)$ subgraphs $\mathcal G_j$'s.

Since $m_{D,\widehat D}(n)$ and $N_{D,\widehat D}(n)$ grow only polynomially in $n$ for fixed $\widehat D,D$, we can choose a large enough $n$, yet independent of $|V|$, such that $\beta_{D,\widehat D}(n)$ is strictly less than $1$, say at most $1/2$. This ensures that the martingale condition is satisfied, and it follows that
\ie
\bar H^2 \ge \frac12 \sum_i \Pi_i^\perp = \frac12 \bar H \implies \gamma(\bar H) \ge \frac12~.
\fe
In other words, $\bar H$ is gapped in the infinite volume limit.

\subsection{Examples}

\begin{figure}[t]
\centering
\hfill \raisebox{-0.5\height}{\includegraphics[scale=0.25]{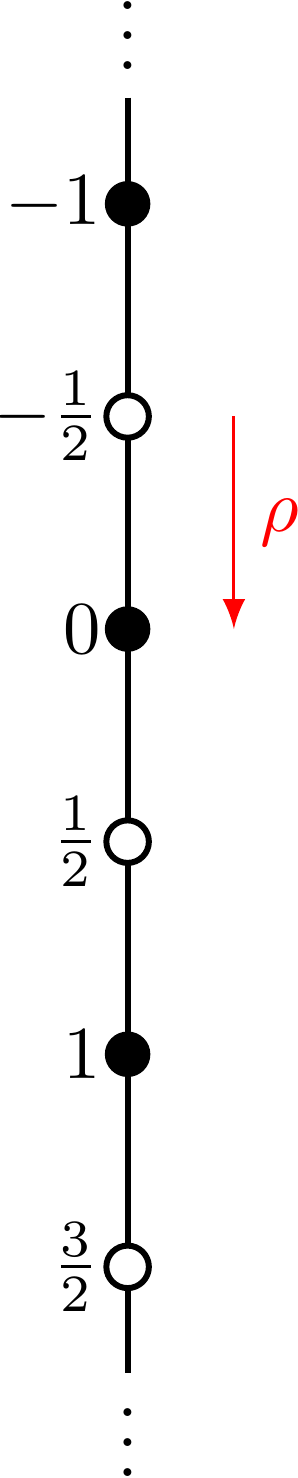}} \hfill \raisebox{-0.5\height}{\includegraphics[scale=0.2]{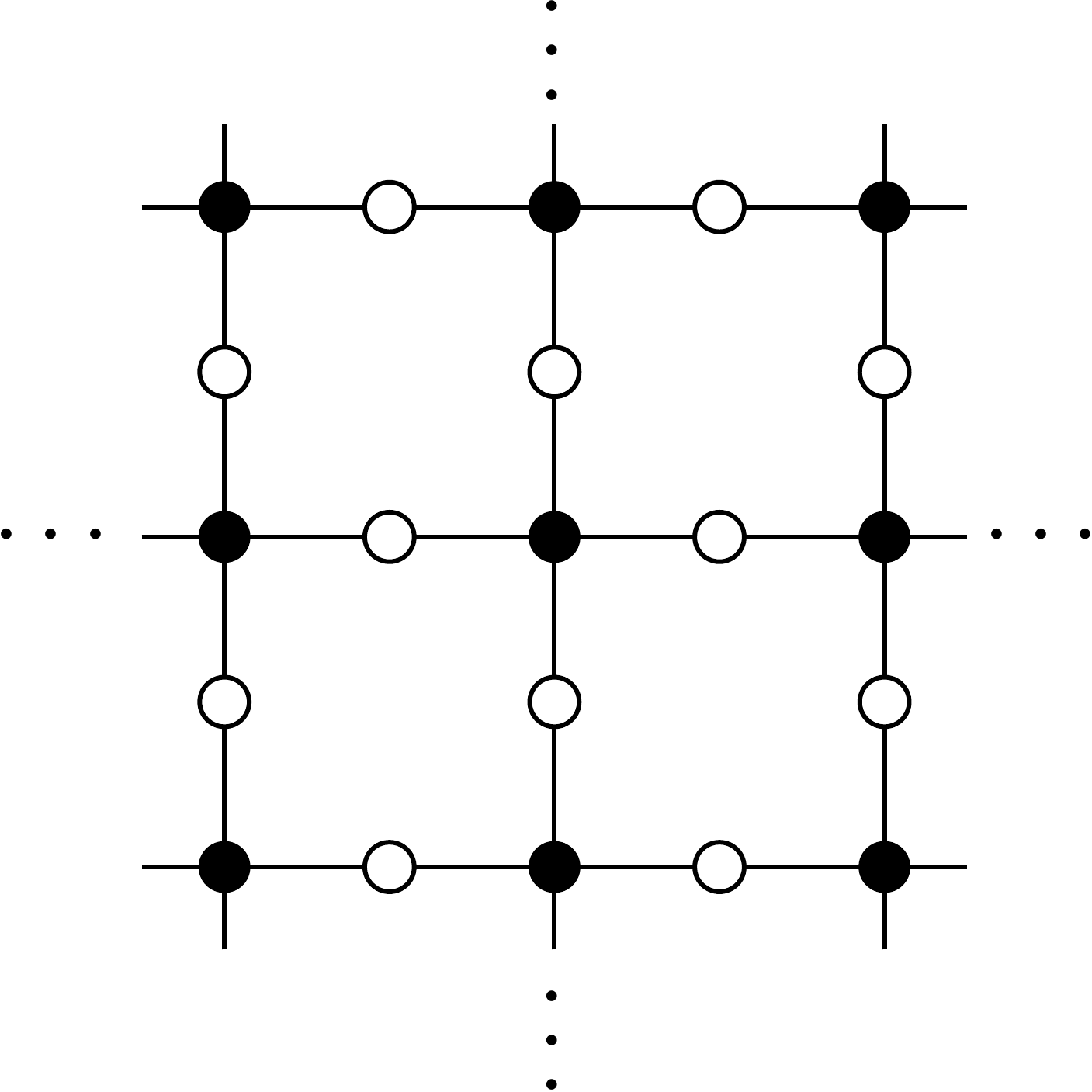}} \hfill \raisebox{-0.5\height}{\includegraphics[scale=0.25]{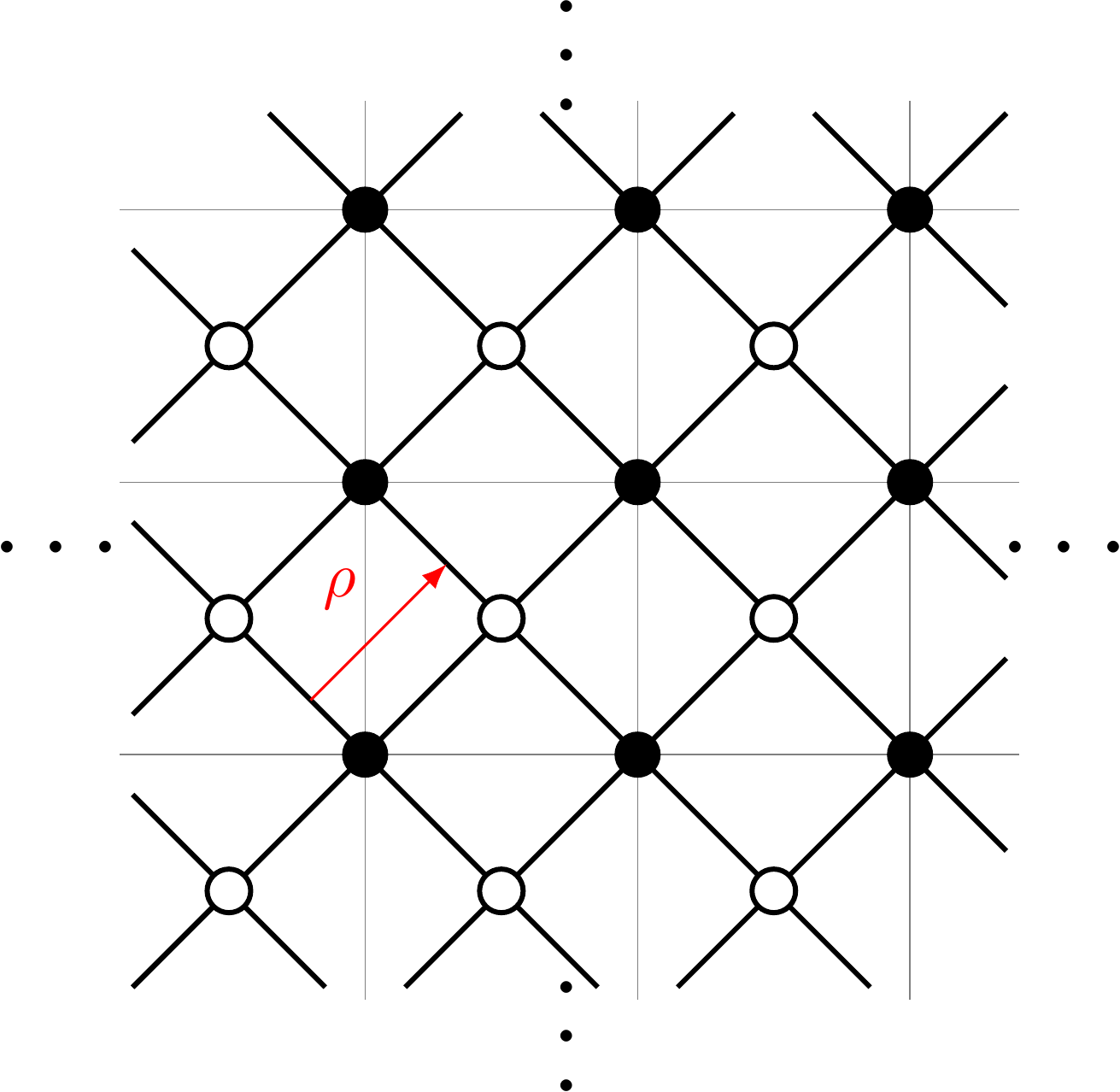}} \hfill \hfill
\\~\\
\hfill (a)~~~~~~~~~~~~~~~~~~\hfill (b) ~~~~~~~~~~~~~~~~~~~~~~~~~~~~~~~~\hfill (c)~~~~~ \hfill \hfill
\caption{The bipartite graphs associated with (a) the 1+1d Ising model, (b) the 2+1d Ising model, and (c) the 2+1d plaquette Ising model. In all cases, the solid (black) vertices represent the physical qubits in $V$ and the hollow (white) vertices represent the interactions in $\widehat V$. In (a) and (c), the red arrow $\rho$ represents a reversing automorphism, given by a ``half-translation'', that generates the non-invertible duality symmetry. The thin grey lines in (c) represent the underlying square lattice. Since the 2+1d Ising model and the 2+1d Toric Code are related via gauging, they have the same bipartite graph, except for exchanging the solid and hollow vertices (and a ``half-translation'' in the (1,1) direction).}\label{fig:gTFIM}
\end{figure}

\subsubsection{1+1d Ising model}
Our first example is a prototypical example with a duality symmetry on the lattice: the 1+1d transverse-field Ising model described by the Hamiltonian
\ie
H_\text{1d-Is} &= - J \sum_i Z_i Z_{i+1} - h \sum_i X_i~,
\fe
on a periodic chain with $L$ sites with a qubit on each site. The associated bipartite graph is simply a cycle on $2L$ vertices, with half of them in $V$ being the sites labelled by $i$ and the other half in $\widehat V$ being the links labelled by $i+\frac12$, as illustrated in Fig. \ref{fig:gTFIM}(a).

This model has a $\mathbb Z_2$ 0-form symmetry generated by $\eta = \prod_i X_i$. It is also invariant under lattice translations and reflections. In addition, there is an obvious reversing automorphism given by the ``half-translation'' of the cycle. It corresponds to the well known Kramers-Wannier duality symmetry that exchanges the Ising and transverse-field terms \cite{Seiberg:2023cdc,Seiberg:2024gek}.

When $h < J$, the $\mathbb Z_2$ 0-form symmetry is spontaneously broken leading to a ferromagnetic phase. In particular, at $h=0$, the two ground states are $|\Zero\> = |0\cdots0\>$ and $|\One\> := |1\cdots1\>$. On the other hand, when $h>J$, the model is in a paramagnetic ($\mathbb Z_2$ symmetric) phase. In particular, at $J=0$, the ground state is $|\Plus\> = |{+}\cdots{+}\>$. The two gapped phases are separated by a second order transition described by the 2d Ising CFT. The transition occurs precisely at $h=J$ because of the Kramers-Wannier duality.

The deformed Hamiltonian is 
\ie
H_\text{1d-Is}^\lambda = \sum_i \left[ - J Z_i Z_{i+1} - h X_i + \frac{J\lambda}2 (X_i Z_{i+1} Z_{i+2} + Z_i Z_{i+1} X_{i+2})\right]~,
\fe
which is precisely the one studied by O'Brien and Fendley in \cite{OBrien:2017wmx}.\footnote{See also \cite{Sannomiya:2017foz} for a similar deformation.} In particular, at the frustration-free point ($J=h,\lambda = 1$) along the deformation, there are three exactly degenerate ground states, $|\Plus\>$, $|\Zero\>$, and $|\One\>$, which remain gapped in the thermodynamic limit \cite{OBrien:2017wmx,Gorantla:2024ptu}. The Kramers-Wannier duality symmetry exchanges $|\Plus\>$ and $|\Zero\>+|\One\>$, so this point along the deformation realises the coexistence of paramagnetic and ferromagnetic phases, where the duality symmetry is spontaneously broken.

\subsubsection{2+1d Ising model}
Our second example is the 2+1d transverse-field Ising model, described by the Hamiltonian
\ie
H_\text{2d-Is} = -J \sum_\ell \left( \begin{tikzpicture}[baseline=0]
\draw (0,0) node[fill=white,inner sep = 2pt]{\footnotesize$Z$} -- (1,0) node[fill=white,inner sep = 2pt]{\footnotesize$Z$};
\end{tikzpicture} + \begin{tikzpicture}[baseline=0]
\draw (0,0) node[fill=white,inner sep = 2pt]{\footnotesize$Z$} -- (0,1) node[fill=white,inner sep = 2pt]{\footnotesize$Z$};
\end{tikzpicture} \right) - h \sum_s X~,
\fe
on a periodic square lattice with $L_x \times L_y$ sites with a qubit on each site. The associated bipartite graph is the Lieb lattice shown in Fig. \ref{fig:gTFIM}(b).

This model has a $\mathbb Z_2$ 0-form symmetry generated by $\eta = \prod_s X_s$. It is also invaraint under lattice translations, rotations, and reflections. However, unlike its 1+1d cousin, there is no reversing automorphism in this case because there are twice as many links in $\widehat V$ as there are sites in $V$, so there is no non-invertible duality symmetry in this model.

Like in the 1+1d case, there are two gapped phases: a ferromagnetic phase with two ground states for $h \ll J$, and a paramagnetic phase with one ground state for $h\gg J$. However, since there is no duality in this case, the transition between them is hard to analyse. Numerics suggests that there is a second-order transition described by the 3d Ising CFT is at $h/J \approx 3.044$ \cite{du_Croo_de_Jongh_1998,Blote:2002ieo}.

Nevertheless, consider deforming $H_\text{2d-Is}$ by the following terms: 
\ie\label{2d-Is-deform}
\begin{tikzpicture}[baseline=0]
\draw (0,0) node[fill=white,inner sep = 2pt]{\footnotesize$Z$} -- (1,0) node[fill=white,inner sep = 2pt]{\footnotesize$Z$} -- (2,0) node[fill=white,inner sep = 2pt]{\footnotesize$X$};
\end{tikzpicture} + \begin{tikzpicture}[baseline=0]
\draw (0,0) node[fill=white,inner sep = 2pt]{\footnotesize$Z$} -- (1,0) node[fill=white,inner sep = 2pt]{\footnotesize$Z$} -- (1,1) node[fill=white,inner sep = 2pt]{\footnotesize$X$};
\end{tikzpicture} + \cdots~,
\fe
where ``$\cdots$'' includes all terms obtained by translating, rotating, and reflecting the shown terms. Overall, for each Ising term, there are two deformation terms of the first type and four of the second. They manifestly preserve the $\mathbb Z_2$ internal symmetry and all the spatial symmetries, i.e., lattice translations, rotations, and reflections. Moreover, there is a frustration-free point along this deformation, with $h=2J$,\footnote{The reason for the factor of $2$ is that there are twice as many Ising terms as the transverse-field terms.} where there are three exactly degenerate ground states, $|\Plus\>$, $|\Zero\>$, and $|\One\>$, which remain gapped in the thermodynamic limit. Therefore, this point realises a coexistence of paramagnetic and ferromagnetic phases of the 2+1d transverse-field Ising model. However, there is no duality symmetry that exchanges these two phases.

\subsubsection{2+1d plaquette Ising model}
Another example, now with fractonic physics, is the 2+1d transverse-field plaquette Ising model, also known as the Xu-Moore model \cite{Xu_2005}, described by the Hamiltonian
\ie
H_\text{2d-plaq-Is} = -J \sum_p \begin{tikzpicture}[baseline=0]
\draw (0,0) node[fill=white,inner sep = 2pt]{\footnotesize$Z$} -- (1,0) node[fill=white,inner sep = 2pt]{\footnotesize$Z$} -- (1,1) node[fill=white,inner sep = 2pt]{\footnotesize$Z$} -- (0,1) node[fill=white,inner sep = 2pt]{\footnotesize$Z$} -- cycle;
\end{tikzpicture} - h \sum_s X~,
\fe
on a periodic square lattice with $L_x \times L_y$ sites with a qubit on each site. The associated bipartite graph is a $45^\circ$-tilted square lattice shown in Fig. \ref{fig:gTFIM}(c).

This model has a $\mathbb Z_2$ subsystem symmetry generated by the rigid line operators $\eta^y_i = \prod_j X_{i,j}$ and $\eta^x_j = \prod_i X_{i,j}$, of which there are $L_x + L_y - 1$ independent generators. It is also invariant under lattice translations, rotations, and reflections. Like in the 1+1d Ising model example, there is a reversing automorphism given by the ``half-translation'' in the $(1,1)$ direction. The corresponding non-invertible duality symmetry has been studied before \cite{Cao:2023doz, Cao:2023rrb,ParayilMana:2024txy,Spieler:2024fby}.

Once again, there are two gapped phases: a ferromagnetic phase with $2^{L_x + L_y -1}$ ground states for $h\ll J$ and a paramagnetic phase with one ground state for $h\gg J$. Due to the non-invertible duality symmetry, there must be at least one transition at $h=J$. Numerics suggests that, in fact, there is a first-order transition at that point \cite{Or_s_2009}.

The deformation terms in this case are
\ie
\begin{tikzpicture}[baseline=0]
\draw (1,1) -- (2,1) node[fill=white,inner sep = 2pt]{\footnotesize$X$};
\draw (0,0) node[fill=white,inner sep = 2pt]{\footnotesize$Z$} -- (1,0) node[fill=white,inner sep = 2pt]{\footnotesize$Z$} -- (1,1) node[fill=white,inner sep = 2pt]{\footnotesize$Z$} -- (0,1) node[fill=white,inner sep = 2pt]{\footnotesize$Z$} -- cycle;
\end{tikzpicture} + \begin{tikzpicture}[baseline=0]
\draw (1,1) -| (2,2) node[fill=white,inner sep = 2pt]{\footnotesize$X$} -| cycle;
\draw (0,0) node[fill=white,inner sep = 2pt]{\footnotesize$Z$} -- (1,0) node[fill=white,inner sep = 2pt]{\footnotesize$Z$} -- (1,1) node[fill=white,inner sep = 2pt]{\footnotesize$Z$} -- (0,1) node[fill=white,inner sep = 2pt]{\footnotesize$Z$} -- cycle;
\end{tikzpicture} + \cdots~,
\fe
where ``$\cdots$'' includes all terms obtained by translating, rotating, and reflecting the shown terms. The two types of terms are mapped to terms of the same type under the duality symmetry. Overall, for each plaquette-Ising term, there are eight deformations terms of the first type and four of the second. At the frustration-free point, with $h=J$, there are $1+2^{L_x + L_y -1}$ exactly degenerate ground states---one of which is the trivial product state $|\Plus\>$ (paramagnetic phase), and $2^{L_x + L_y -1}$ are the ground states at $h=0$ (ferromagnetic phase) of the 2+1d plaquette Ising model---and they remain gapped in the thermodynamic limit. As expected, the non-invertible duality symmetry exchanges these two phases, so it is spontaneously broken.

\section{Deformed quantum CSS code in transverse field}\label{sec:css}

So far, our discussion was about models based on classical LDPC codes, which realise coexistence of trivial and symmetry breaking phases. In this section, we extend this discussion to a much larger class of models based on quantum CSS codes. They realise coexistence of trivial and topologically ordered phases, including fracton phases. In fact, they realise a gapped phase where non-invertible (Wegner-like) duality symmetries are spontaneously broken. The details are identical to the ones above, so we will be brief here.

Consider another bipartite graph $\widetilde{\mathcal G}$ with vertex-set $\widetilde V \sqcup V$, where $V$ is the same as before, and biadjacency matrix $\tilde{\mathfrak h}$. As before, we assume that $\widetilde{\mathcal G}$ is connected, and both left- and right-degrees are bounded. In other words, it corresponds to another classical LDPC code. In addition, we assume that $\tilde{\mathfrak h}$ satisfies
\ie\label{css-matrix-cond}
\mathfrak h \tilde{\mathfrak h}^\intercal = 0~.
\fe
That is, the columns of $\tilde{\mathfrak h}^\intercal$ are in the kernel of $\mathfrak h$. A pair of classical codes, $\mathcal G$ and $\widetilde{\mathcal G}$, related in this way form a quantum CSS code \cite{Breuckmann:2021yvk}.

Let us incorporate this new information into the quantum Hamiltonian. For each $\tilde v \in \widetilde V$, define the operators
\ie
G_{\tilde v} := \prod_{v\in V} X_v^{\tilde{\mathfrak h}_{\tilde v,v}} = \prod_{v\in V\,:\,v\sim \tilde v} X_v~. 
\fe
The condition \eqref{css-matrix-cond} implies that the $G_{\tilde v}$'s commute with every term of the Hamiltonian \eqref{H-gTFIM}, i.e., they are $\mathbb Z_2^\nu$ symmetry operators. We add these terms to the Hamiltonian \eqref{H-gTFIM} to get a new Hamiltonian
\ie
H_\text{new} &= H - \sum_{\tilde v \in \widetilde V} g_{\tilde v} G_{\tilde v}
\\
&= - \sum_{\hat v\in\widehat V} J_{\hat v} F_{\hat v} - \sum_{v\in V} h_v X_v - \sum_{\tilde v \in \widetilde V} g_{\tilde v} G_{\tilde v}~,
\fe
where $g_{\tilde v}$ is a new coupling constant. The last term penalises the states that are not invariant under $G_{\tilde v}$'s. This is reminiscent of imposing the Gauss law energetically in lattice $\mathbb Z_2$ gauge theory. Indeed, our notation is here suggestive: the $F_{\hat v}$ term is the magnetic ``f''lux, the single $X_v$ term is the electric field, and the $G_{\tilde v}$ term is the ``G''auss law.

The internal $\mathbb Z_2^\nu$ symmetries and the non-invertible duality symmetries of $H$ continue to be symmetries of $H_\text{new}$. In constrast, the spatial symmetries may not be preserved due to the last term, unless the graph $\widetilde{\mathcal G}$ is chosen appropriately.\footnote{For instance, one can take $\widetilde{\mathcal G}$ to be the bipartite graph associated with all $\mathbb Z_2^\nu$ symmetry operators $\eta_a$ that are supported on at most $\widetilde D$ qubits, where $\widetilde D$ is a positive integer. Any spatial symmetry permutes the $\mathbb Z_2^\nu$ symmetry operators according to the top-left equation of \eqref{gTFIM-fusion-alg}, so the above collection of operators with finite support is preserved under spatial symmetries. One can then take all $g_{\tilde v}$ to be equal so that $H_\text{new}$ is invariant under the spatial symmetries of $H$.}

When $J_{\hat v} = 0$, $h_v >0$, and $g_{\tilde v} \ge 0$, the only ground state of $H_\text{new}$ is $|\Plus\>$. Borrowing terminology from lattice $\mathbb Z_2$ gauge theory, we refer to this phase as the confining or trivial or $\mathbb Z_2^\nu$-preserving phase. On the other hand, when $h_v = 0$ and $J_{\hat v},g_{\tilde v} >0$, the Hamiltonian is a commuting projector model based on the quantum CSS code. Its ground states are the simultaneous eigenstates of $F_{\hat v}$'s and $G_{\tilde v}$'s with eigenvalue $1$, which form the \emph{code space}. More concretely, they are
\ie
\prod_{\tilde v \in \widetilde V} \left( \frac{1+G_{\tilde v}}2 \right) |a\>~,
\fe
where $|a\>$ is the ground state of $H$ at $h_v=0$ defined in \eqref{graph-eta-state}. Note that, because of the projection, not all $a \in \ker \mathfrak h$ give distinct states. In particular, if $\eta_a$ and $\eta_b$ differ by product of some $G_{\tilde v}$'s, then $|a\>$ and $|b\>$ are projected to the same ground state of $H_\text{new}$. The other ground states of $H$ that do not preserve $G_{\tilde v}$'s are pushed to energies proportional $g_{\tilde v}$'s. We refer to this phase as the Higgs or topologically-ordered or $\mathbb Z_2^\nu$-broken phase.

Like in Section \ref{sec:graph-deformed}, one can consider a deformation of $H_\text{new}$ that preserves all of its symmetries. We are particularly interested in the analog of \eqref{graph-deformedH'}:
\ie\label{css-deformedH'}
H'_\text{new} = H' + \sum_{\tilde v \in \widetilde V} g_{\tilde v} (1-G_{\tilde v}) = J \sum_{v,\hat v\,:\, d(v,\hat v)=3} P_v Q_{\hat v} + \sum_{\tilde v \in \widetilde V} g_{\tilde v} (1-G_{\tilde v})~,
\fe
where we made the trivial change $G_{\tilde v} \to -(1-G_{\tilde v})$ to ensure that $H'_\text{new}$ is positive semi-definite, i.e., $H'_\text{new} \ge 0$. It is easy to check that
\ie
|\Plus\>~,\quad \text{and} \quad \prod_{\tilde v \in \widetilde V} \left( \frac{1+G_{\tilde v}}2 \right) |a\>~,\quad a\in \ker \mathfrak h~,
\fe
are zero energy ground states of $H'_\text{new}$, which makes $H'_\text{new}$ frustration-free. These are precisely the ground states of $H'$ projected onto the subspace preserving $G_{\tilde v}$'s. In fact, the proof that there are no other ground states is the same as the one for $H'$. Furthermore, as long as $g_{\tilde v}$'s are larger than the gap of $H'$, the proof of gap of $H'_\text{new}$ in the thermodynamic limit also goes through unchanged.

If the original Hamiltonian $H_\text{new}$ has a duality symmetry, then the deformed Hamiltonian $H'_\text{new}$ preserves it too. The action of the duality operators on the above ground states can be inferred from \eqref{gTFIM-D-action}. In particular, the duality symmetry exchanges the trivial and topologically-ordered phases. Therefore, the Hamiltonian $H'_\text{new}$ provides a concrete realisation of a gapped phase where the duality symmetry is spontaneously broken.

\subsection{Examples}

\subsubsection{2+1d Toric Code}
Our first example is a prototypical example of topological order: the 2+1d Toric Code (in the presence of transverse field for $h>0$) described by the Hamiltonian
\ie
H_\text{2d-TC} = -J \sum_p \begin{tikzpicture}[baseline=0]
\draw (0,0) -- node[fill=white,inner sep = 2pt]{\footnotesize$Z$} (1,0) -- node[fill=white,inner sep = 2pt]{\footnotesize$Z$} (1,1) -- node[fill=white,inner sep = 2pt]{\footnotesize$Z$} (0,1) -- node[fill=white,inner sep = 2pt]{\footnotesize$Z$} cycle;
\end{tikzpicture} - h \sum_\ell \left( \begin{tikzpicture}[baseline=0]
\draw (0,0) -- node[fill=white,inner sep = 2pt]{\footnotesize$X$} (1,0);
\end{tikzpicture} + \begin{tikzpicture}[baseline=0]
\draw (0,0) -- node[fill=white,inner sep = 2pt]{\footnotesize$X$} (0,1);
\end{tikzpicture} \right) - g \sum_s \begin{tikzpicture}[baseline=0]
\draw (-1,0) -- node[fill=white,inner sep = 2pt]{\footnotesize$X$} (0,0) -- node[fill=white,inner sep = 2pt]{\footnotesize$X$} (1,0) (0,-1) -- node[fill=white,inner sep = 2pt]{\footnotesize$X$} (0,0) -- node[fill=white,inner sep = 2pt]{\footnotesize$X$} (0,1);
\end{tikzpicture}~,
\fe
on a periodic square lattice with $L_x \times L_y$ sites with a qubit on each link. Equivalently, this is the Hamiltonian of the 2+1d lattice $\mathbb Z_2$ gauge theory, where the first term is the magnetic flux, the second is the electric field, and the third is the Gauss law. The associated bipartite graph is the same as the one for the 2+1d Ising model, i.e., the Lieb lattice shown in Fig. \ref{fig:gTFIM}(b), with $V$ and $\widehat V$ exchanged. This is a consequence of the fact that the two models are related via gauging.

This model has a $\mathbb Z_2$ 1-form symmetry generated by the line operators $\eta(\hat \gamma) = \prod_{\ell \in \hat \gamma} X_\ell$, where $\hat \gamma$ is a dual curve on the dual lattice, and $\ell \in \hat \gamma$ means the link $\ell$ pierces $\hat \gamma$. This includes the right-most terms in the Hamiltonian, i.e., the vertex (Gauss law) terms. It is also invariant under lattice translations, rotations, and reflections. However, like the 2+1d Ising model, this model does not have a non-invertible duality symmetry.\footnote{When $h=0$ and $J=g$, this model has a different symmetry known as the $em$-duality that exchanges the first and the third terms. This is an invertible symmetry and is not related to gauging the $\mathbb Z_2$ 1-form symmetry. More generally, when the two graphs $\mathcal G$ and $\widetilde{\mathcal G}$ are isomorphic, the resulting quantum CSS code is self-dual in the absence of the transverse field.}

When $h\ll J$, the model is topologically ordered with four ground states, i.e., the $\mathbb Z_2$ 1-form symmetry is spontaneously broken. In particular, at $h=0$, the ground states are precisely the 2+1d Toric Code ground states. On the other hand, when $h\gg J$, there is only one ground state corresponding to a trivial phase. In particular, at $J=0$, the ground state is the trivial product state $|\Plus\>$. Since this model is related to the 2+1d Ising model via gauging a discrete symmetry, their phase diagrams are in one-one correspondence except for the exchange $h\leftrightarrow J$. In particular, the transition is second-order, described by the 3d Ising* CFT.

The deformations terms are
\ie\label{2d-TC-deform}
\begin{tikzpicture}[baseline=0]
\draw (1,1) -- node[fill=white,inner sep = 2pt]{\footnotesize$X$} (2,1);
\draw (0,0) -- node[fill=white,inner sep = 2pt]{\footnotesize$Z$} (1,0) -- node[fill=white,inner sep = 2pt]{\footnotesize$Z$} (1,1) -- node[fill=white,inner sep = 2pt]{\footnotesize$Z$} (0,1) -- node[fill=white,inner sep = 2pt]{\footnotesize$Z$} cycle;
\end{tikzpicture} + \begin{tikzpicture}[baseline=0]
\draw (1,1) -- (2,1) -- node[fill=white,inner sep = 2pt]{\footnotesize$X$} (2,0) -| cycle;
\draw (0,0) -- node[fill=white,inner sep = 2pt]{\footnotesize$Z$} (1,0) -- node[fill=white,inner sep = 2pt]{\footnotesize$Z$} (1,1) -- node[fill=white,inner sep = 2pt]{\footnotesize$Z$} (0,1) -- node[fill=white,inner sep = 2pt]{\footnotesize$Z$} cycle;
\end{tikzpicture} + \cdots~,
\fe
where ``$\cdots$'' includes all terms obtained by translating, rotating, and reflecting the shown terms. Under gauging the $\mathbb Z_2$ 1-form symmetry, the two types of terms are mapped to the two types of terms in \eqref{2d-Is-deform}. Overall, for each plaquette (magnetic flux) term, there are eight deformation terms of the first type and four of the second. At the frustration free point, with $h=J/2$,\footnote{The reason for the factor of $2$ is that there are twice as many link (electric field) terms as the plaquette (magnetic flux) terms.} there are five exactly degenerate ground states, one of which is the trivial product state $|\Plus\>$ and the rest are the four Toric-Code ground states, and they remain gapped in the thermodynamic limit. Therefore, this point realises the coexistence of the trivial and topologically ordered phases. However, there is no duality symmetry that exchanges the two phases.

\begin{figure}
\centering
\hfill \raisebox{-0.5\height}{\includegraphics[scale=0.16]{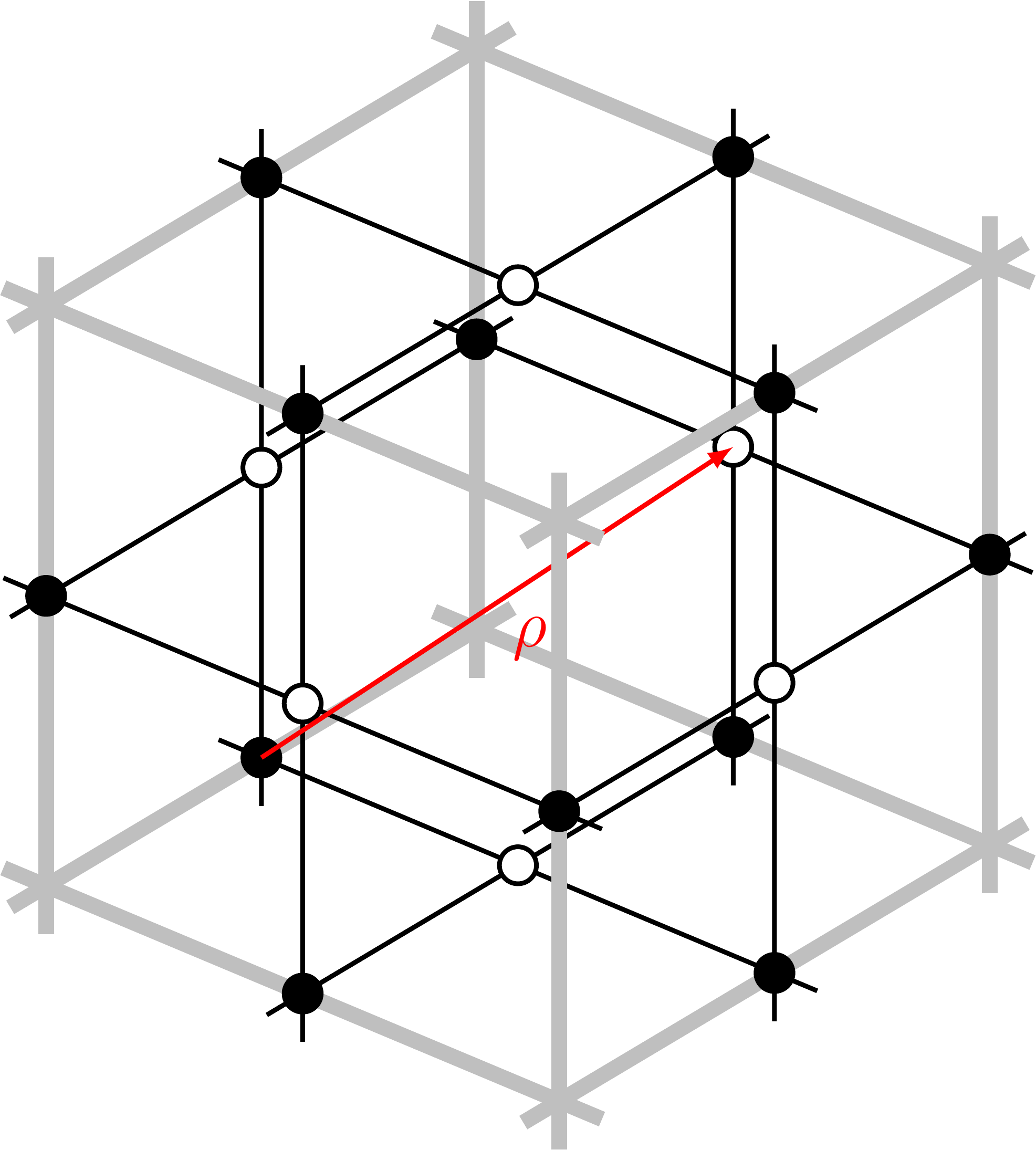}} \hfill \raisebox{-0.5\height}{\includegraphics[scale=0.16]{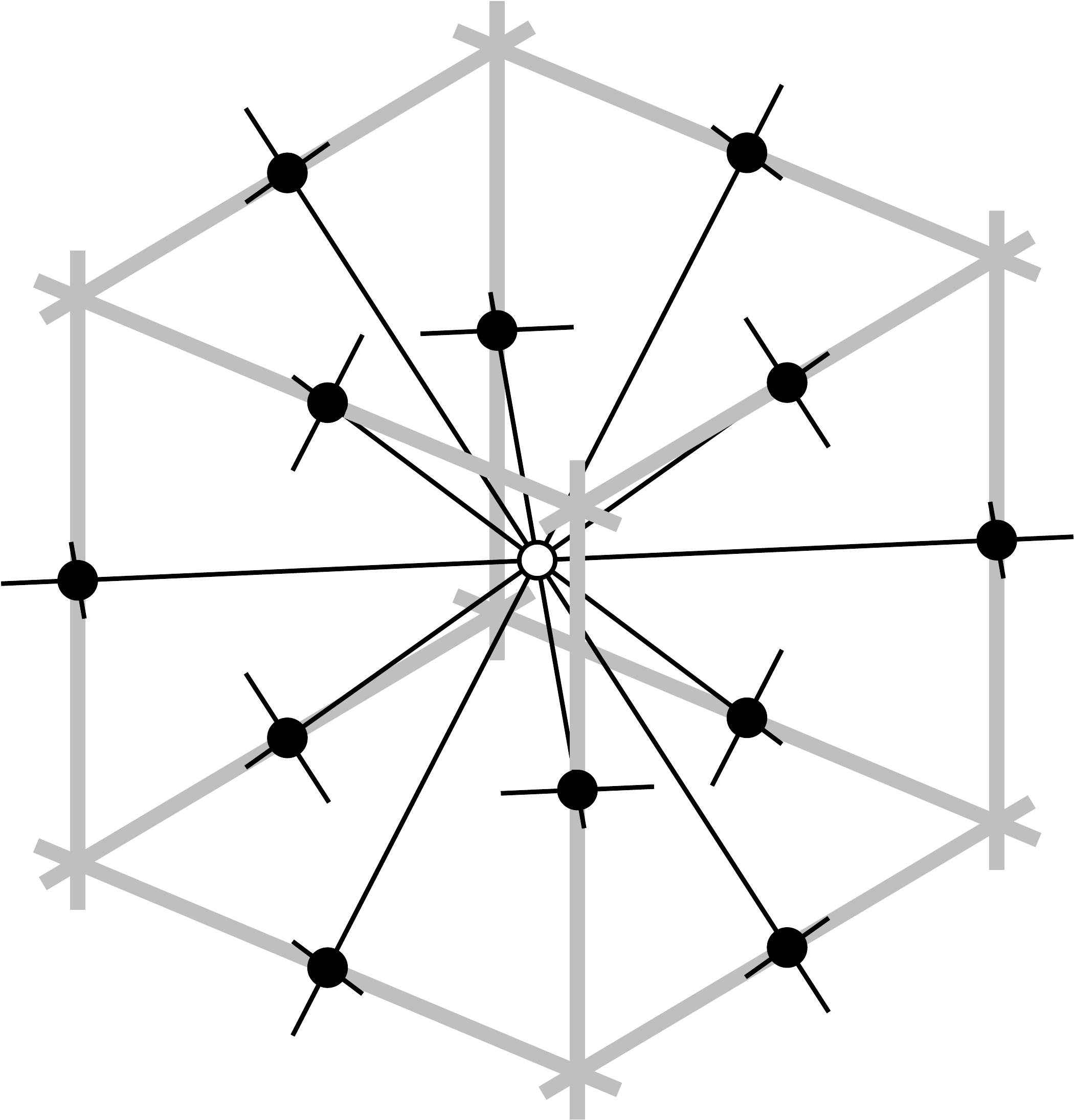}} \hfill ~
\\~\\
\hfill (a)~~~~~~~~~~\hfill ~~~~~~~~~~(b)\hfill ~~~
\caption{The bipartite graphs associated with (a) the 3+1d Toric Code and (b) the X-Cube model. In both cases, the solid (black) vertices represent the physical qubits in $V$ and the hollow (white) vertices represent the interactions in $\widehat V$. In (a), the red arrow $\rho$ represents a reversing automorphism, given by a ``half-translation'' in the $(1,1,1)$ direction, that generates the non-invertible Wegner duality symmetry. The thick grey lines represent the underlying cubic lattice.}
\label{fig:CSS}
\end{figure}

\subsubsection{3+1d Toric Code}
Our next example is the 3+1d Toric Code (in the presence of a transverse field for $h>0$) described by the Hamiltonian
\tdplotsetmaincoords{60}{130}
\ie
H_\text{3d-TC} &= -J \sum_p \left( \begin{tikzpicture}[tdplot_main_coords,scale=1.25,baseline=0]
\draw (0,0,0) -- node[fill=white,inner sep = 1pt]{\footnotesize$Z$} (1,0,0) -- node[fill=white,inner sep = 1pt]{\footnotesize$Z$} (1,0,1) -- node[fill=white,inner sep = 1pt]{\footnotesize$Z$} (0,0,1) -- node[fill=white,inner sep = 1pt]{\footnotesize$Z$} cycle;
\end{tikzpicture} + \begin{tikzpicture}[tdplot_main_coords,scale=1.25,baseline=0]
\draw (0,0,0) -- node[fill=white,inner sep = 1pt]{\footnotesize$Z$} (1,0,0) -- node[fill=white,inner sep = 1pt]{\footnotesize$Z$} (1,1,0) -- node[fill=white,inner sep = 1pt]{\footnotesize$Z$} (0,1,0) -- node[fill=white,inner sep = 1pt]{\footnotesize$Z$} cycle;
\end{tikzpicture} +
\begin{tikzpicture}[tdplot_main_coords,scale=1.25,baseline=0]
\draw (0,0,0) -- node[fill=white,inner sep = 1pt]{\footnotesize$Z$} (0,1,0) -- node[fill=white,inner sep = 1pt]{\footnotesize$Z$} (0,1,1) -- node[fill=white,inner sep = 1pt]{\footnotesize$Z$} (0,0,1) -- node[fill=white,inner sep = 1pt]{\footnotesize$Z$} cycle;
\end{tikzpicture} \right)
\\
&\quad - h \sum_\ell \left( \begin{tikzpicture}[tdplot_main_coords,scale=1.25,baseline=0]
\draw (0,0,0) -- node[fill=white,inner sep = 1pt]{\footnotesize$X$} (0,1,0);
\end{tikzpicture}+\begin{tikzpicture}[tdplot_main_coords,scale=1.25,baseline=0]
\draw (0,0,0) -- node[fill=white,inner sep = 1pt]{\footnotesize$X$} (0,0,1);
\end{tikzpicture}+
\begin{tikzpicture}[tdplot_main_coords,scale=1.25,baseline=0]
\draw (0,0,0) -- node[fill=white,inner sep = 1pt]{\footnotesize$X$} (1,0,0);
\end{tikzpicture} \right) - g \sum_s \begin{tikzpicture}[tdplot_main_coords,scale=1.25,baseline=0]
\draw (-1,0,0) -- node[fill=white,inner sep = 1pt]{\footnotesize$X$} (0,0,0) -- node[fill=white,inner sep = 1pt]{\footnotesize$X$} (1,0,0);
\draw (0,-1,0) -- node[fill=white,inner sep = 1pt]{\footnotesize$X$} (0,0,0) -- node[fill=white,inner sep = 1pt]{\footnotesize$X$} (0,1,0);
\draw (0,0,-1) -- node[fill=white,inner sep = 1pt]{\footnotesize$X$} (0,0,0) -- node[fill=white,inner sep = 1pt]{\footnotesize$X$} (0,0,1);
\end{tikzpicture}~,
\fe
on a periodic cubic lattice with $L_x \times L_y \times L_z$ sites with a qubit on each link. (We show only one orientation of plaquette and link terms for brevity.) Equivalently, this is the Hamiltonian of the 3+1d lattice $\mathbb Z_2$ gauge theory, where the first term is the magnetic flux, the second is the electric field, and the third is the Gauss law. The associated bipartite graph is the cubic lattice with edge centres in $V$ and face centres in $\widehat V$, as shown in Fig. \ref{fig:CSS}(a).

Like the 2+1d Toric Code, this model has a $\mathbb Z_2$ 1-form symmetry generated by the surface operators $\eta(\widehat \Sigma) = \prod_{\ell \in \widehat \Sigma} X_\ell$, where $\widehat \Sigma$ is a dual curve on the dual lattice, and $\ell \in \widehat \Sigma$ means the link $\ell$ pierces $\widehat \Sigma$. This includes the right-most terms in the Hamiltonian, i.e., the vertex (Gauss law) terms. It is also invariant under lattice translations, rotations, and reflections. Unlike its 2+1d cousin, this model does have a non-invertible duality symmetry associated with the reversing automorphism given by the ``half-translation'' in the $(1,1,1)$ direction. It is the well known Wegner duality symmetry, which exchanges the magnetic flux and the electric field terms \cite{Koide:2021zxj,Choi:2021kmx,Kaidi:2021xfk,Gorantla:2024ocs}.

There are two gapped phases: a topologically ordered phase with eight ground states for $h\ll J$ and a trivial phase with one ground state for $h\gg J$. They must be separated by at least one transition at $h=J$ due to the Wegner duality symmetry. Numerics suggests that this transition is of first-order \cite{PhysRevLett.42.1390,PhysRevD.20.1915}.

The deformation terms are
\ie
\begin{tikzpicture}[tdplot_main_coords,scale=1.25,baseline=0]
\draw (0,0,0) -- node[fill=white,inner sep = 1pt]{\footnotesize$Z$} (0,1,0) -- node[fill=white,inner sep = 1pt]{\footnotesize$Z$} (0,1,1) -- node[fill=white,inner sep = 1pt]{\footnotesize$Z$} (0,0,1) -- node[fill=white,inner sep = 1pt]{\footnotesize$Z$} cycle;
\draw (0,1,1) -- node[fill=white,inner sep = 1pt]{\footnotesize$X$} (-1,1,1);
\end{tikzpicture} ~ +~  \begin{tikzpicture}[tdplot_main_coords,scale=1.25,baseline=0]
\draw (0,0,0) -- node[fill=white,inner sep = 1pt]{\footnotesize$Z$} (0,1,0) -- node[fill=white,inner sep = 1pt]{\footnotesize$Z$} (0,1,1) -- node[fill=white,inner sep = 1pt]{\footnotesize$Z$} (0,0,1) -- node[fill=white,inner sep = 1pt]{\footnotesize$Z$} cycle;
\draw (0,1,1) -- node[fill=white,inner sep = 1pt]{\footnotesize$X$} (0,1,2);
\end{tikzpicture} ~+~ \begin{tikzpicture}[tdplot_main_coords,scale=1.25,baseline=0]
\draw (0,0,0) -- node[fill=white,inner sep = 1pt]{\footnotesize$Z$} (0,1,0) -- node[fill=white,inner sep = 1pt]{\footnotesize$Z$} (0,1,1) -- node[fill=white,inner sep = 1pt]{\footnotesize$Z$} (0,0,1) -- node[fill=white,inner sep = 1pt]{\footnotesize$Z$} cycle;
\draw (0,0,1) -- (0,0,2) -- node[fill=white,inner sep = 1pt]{\footnotesize$X$} (0,1,2) -- (0,1,1);
\end{tikzpicture} ~+~ \begin{tikzpicture}[tdplot_main_coords,scale=1.25,baseline=0]
\draw (0,0,0) -- node[fill=white,inner sep = 1pt]{\footnotesize$Z$} (0,1,0) -- node[fill=white,inner sep = 1pt]{\footnotesize$Z$} (0,1,1) -- node[fill=white,inner sep = 1pt]{\footnotesize$Z$} (0,0,1) -- node[fill=white,inner sep = 1pt]{\footnotesize$Z$} cycle;
\draw (0,0,1) -- (-1,0,1) -- node[fill=white,inner sep = 1pt]{\footnotesize$X$} (-1,1,1) -- (0,1,1);
\end{tikzpicture} ~+ \cdots~,
\fe
where ``$\cdots$'' includes all terms obtained by translating, rotating, and reflecting the shown terms. Terms of the first and third types are mapped to terms of the same type under the Wegner duality, whereas terms of the second and fourth types are exchanged. Overall, for each plaquette (magnetic flux) term, there are eight deformation terms of the first type, eight of the second, four of the third, and eight of the fourth. At the frustration-free point, with $h=J$, there are nine exactly degenerate ground states, one of which is the product state $|\Plus\>$ and the rest are the 3+1d Toric Code ground states, and they remain gapped in the thermodynamic limit. Therefore, this point realises the coexistence of the trivial phase and a topologically ordered phase. Moreover, the Wegner duality symmetry exchanges the two phases, so it is spontaneously broken. In a recent work \cite{Gorantla:2024ptu}, we considered only deformations of the first type, and even then we showed the existence of a frustration-free point with the same properties.

\subsubsection{X-Cube model}
Here is another example of topological order, now with fractonic physics: the X-Cube model \cite{Vijay:2016phm} (in the presence of a transverse field for $h>0$) described by the Hamiltonian\footnote{In the original X-Cube model, the $Z$'s and $X$'s are exchanged, so the cube term is a product of $X$'s instead of $Z$'s, leading to the name ``X-Cube''.}
\ie
H_\text{XC} &= -J \sum_c \begin{tikzpicture}[tdplot_main_coords,scale=1.25,baseline=0]
\foreach \x in {0,1} {
  \foreach \y in {0,1} {
    \draw (\x,\y,0) -- node[fill=white,inner sep = 1pt]{\footnotesize$Z$} (\x,\y,1);
  }
}

\foreach \x in {0,1} {
  \foreach \y in {0,1} {
    \draw (\x,0,\y) -- node[fill=white,inner sep = 1pt]{\footnotesize$Z$} (\x,1,\y);
  }
}

\foreach \x in {0,1} {
  \foreach \y in {0,1} {
    \draw (0,\x,\y) -- node[fill=white,inner sep = 1pt]{\footnotesize$Z$} (1,\x,\y);
  }
}
\end{tikzpicture} - h \sum_\ell \left( \begin{tikzpicture}[tdplot_main_coords,scale=1.25,baseline=0]
\draw (0,0,0) -- node[fill=white,inner sep = 1pt]{\footnotesize$X$} (0,1,0);
\end{tikzpicture}+\begin{tikzpicture}[tdplot_main_coords,scale=1.25,baseline=0]
\draw (0,0,0) -- node[fill=white,inner sep = 1pt]{\footnotesize$X$} (0,0,1);
\end{tikzpicture}+
\begin{tikzpicture}[tdplot_main_coords,scale=1.25,baseline=0]
\draw (0,0,0) -- node[fill=white,inner sep = 1pt]{\footnotesize$X$} (1,0,0);
\end{tikzpicture} \right)
\\
&\quad - g \sum_s \left( \begin{tikzpicture}[tdplot_main_coords,scale=1.25,baseline=0]
\draw (-1,0,0) -- node[fill=white,inner sep = 1pt]{\footnotesize$X$} (0,0,0) -- node[fill=white,inner sep = 1pt]{\footnotesize$X$} (1,0,0);
\draw (0,0,-1) -- node[fill=white,inner sep = 1pt]{\footnotesize$X$} (0,0,0) -- node[fill=white,inner sep = 1pt]{\footnotesize$X$} (0,0,1);
\end{tikzpicture} + \begin{tikzpicture}[tdplot_main_coords,scale=1.25,baseline=0]
\draw (-1,0,0) -- node[fill=white,inner sep = 1pt]{\footnotesize$X$} (0,0,0) -- node[fill=white,inner sep = 1pt]{\footnotesize$X$} (1,0,0);
\draw (0,-1,0) -- node[fill=white,inner sep = 1pt]{\footnotesize$X$} (0,0,0) -- node[fill=white,inner sep = 1pt]{\footnotesize$X$} (0,1,0);
\end{tikzpicture} + 
\begin{tikzpicture}[tdplot_main_coords,scale=1.25,baseline=0]
\draw (0,-1,0) -- node[fill=white,inner sep = 1pt]{\footnotesize$X$} (0,0,0) -- node[fill=white,inner sep = 1pt]{\footnotesize$X$} (0,1,0);
\draw (0,0,-1) -- node[fill=white,inner sep = 1pt]{\footnotesize$X$} (0,0,0) -- node[fill=white,inner sep = 1pt]{\footnotesize$X$} (0,0,1);
\end{tikzpicture} \right)
\fe
on a periodic cubic lattice with $L_x \times L_y \times L_z$ sites with a qubit on each link. The associated bipartite graph is the cubic lattice with edge centres in $V$ and body centres in $\widehat V$, as shown in Fig. \ref{fig:CSS}(b).

This model has a planar $\mathbb Z_2$ subsystem symmetry generated by the (rigid) operators $\eta^x_{j,k+\frac12} = \prod_i X_{i,j,k+\frac12}$ and similar operators in the other directions. It is also invariant under lattice translations, rotations, and reflections. However, there is no reversing automorphism in this case because $V$ has thrice as many vertices as $\widehat V$, or equivalently, there are thrice as many qubits (on links) as there are X-Cube terms (on cubes). Therefore, this model has no non-invertible duality symmetry.

Once again, there are two gapped phases: a (foliated or type I) fracton phase with $2^{2L_x + 2L_y + 2L_z -3}$ ground states for $h\ll J$ and a trivial phase with one ground state for $h\gg J$. Numerics suggests that they are separated by a first-order transition at around $h/J \approx 0.293$ \cite{Devakul:2017ays,Muhlhauser:2019rjg,Zhou:2022ebw}.

Nevertheless, consider deforming $H_\text{XC}$ by the following terms:
\ie
\begin{tikzpicture}[tdplot_main_coords,scale=1.25,baseline=0]
\foreach \x in {0,1} {
  \foreach \y in {0,1} {
    \draw (\x,\y,0) -- node[fill=white,inner sep = 1pt]{\footnotesize$Z$} (\x,\y,1);
  }
}

\foreach \x in {0,1} {
  \foreach \y in {0,1} {
    \draw (\x,0,\y) -- node[fill=white,inner sep = 1pt]{\footnotesize$Z$} (\x,1,\y);
  }
}

\foreach \x in {0,1} {
  \foreach \y in {0,1} {
    \draw (0,\x,\y) -- node[fill=white,inner sep = 1pt]{\footnotesize$Z$} (1,\x,\y);
  }
}

\draw (0,0,1) -- node[fill=white,inner sep = 1pt]{\footnotesize$X$} (0,0,2);
\end{tikzpicture} ~+~ \begin{tikzpicture}[tdplot_main_coords,scale=1.25,baseline=0]
\foreach \x in {0,1} {
  \foreach \y in {0,1} {
    \draw (\x,\y,0) -- node[fill=white,inner sep = 1pt]{\footnotesize$Z$} (\x,\y,1);
  }
}

\foreach \x in {0,1} {
  \foreach \y in {0,1} {
    \draw (\x,0,\y) -- node[fill=white,inner sep = 1pt]{\footnotesize$Z$} (\x,1,\y);
  }
}

\foreach \x in {0,1} {
  \foreach \y in {0,1} {
    \draw (0,\x,\y) -- node[fill=white,inner sep = 1pt]{\footnotesize$Z$} (1,\x,\y);
  }
}

\draw (0,0,1) -- (0,0,2) -- node[fill=white,inner sep = 1pt]{\footnotesize$X$} (0,1,2) -- (0,1,1);
\end{tikzpicture} ~+~ \begin{tikzpicture}[tdplot_main_coords,scale=1.25,baseline=0]
\foreach \x in {0,1} {
  \foreach \y in {0,1} {
    \draw (\x,\y,0) -- node[fill=white,inner sep = 1pt]{\footnotesize$Z$} (\x,\y,1);
  }
}

\foreach \x in {0,1} {
  \foreach \y in {0,1} {
    \draw (\x,0,\y) -- node[fill=white,inner sep = 1pt]{\footnotesize$Z$} (\x,1,\y);
  }
}

\foreach \x in {0,1} {
  \foreach \y in {0,1} {
    \draw (0,\x,\y) -- node[fill=white,inner sep = 1pt]{\footnotesize$Z$} (1,\x,\y);
  }
}

\draw (0,0,1) -- (0,0,2) -- node[fill=white,inner sep = 1pt]{\footnotesize$X$} (-1,0,2) -- (-1,0,1) -- cycle;
\end{tikzpicture} ~+~ \begin{tikzpicture}[tdplot_main_coords,scale=1.25,baseline=0]
\foreach \x in {0,1} {
  \foreach \y in {0,1} {
    \draw (\x,\y,0) -- node[fill=white,inner sep = 1pt]{\footnotesize$Z$} (\x,\y,1);
  }
}

\foreach \x in {0,1} {
  \foreach \y in {0,1} {
    \draw (\x,0,\y) -- node[fill=white,inner sep = 1pt]{\footnotesize$Z$} (\x,1,\y);
  }
}

\foreach \x in {0,1} {
  \foreach \y in {0,1} {
    \draw (0,\x,\y) -- node[fill=white,inner sep = 1pt]{\footnotesize$Z$} (1,\x,\y);
  }
}

\draw (0,0,1) -- (-1,0,1) -- (-1,0,2) -- node[fill=white,inner sep = 1pt]{\footnotesize$X$} (-1,1,2) -- (-1,1,1) -- (0,1,1);
\draw (-1,0,1) -- (-1,1,1);
\end{tikzpicture} ~+ \cdots~,
\fe
where ``$\cdots$'' includes all terms obtained by translating, rotating, and reflecting the shown terms. Overall, for each cube term, there are twenty-four terms of the first type, twenty-four of the second, forty-eight of the third, and twelve of the fourth. At the frustration-free point, with $h=J/3$,\footnote{The reason for the factor of $3$ is that there are thrice as many link terms as there are cube terms.} there are $1+2^{2L_x + 2L_y + 2L_z -3}$ ground states, one of which is the trivial product state $|\Plus\>$ and the rest are the X-Cube ground states, and they remain gapped in the thermodynamic limit. Therefore, this point realises the coexistence of the trivial phase and the (foliated or type I) fracton order of the X-Cube model. However, there is no non-invertible duality symmetry that exchanges these two phases.

Surprisingly, the value of $h/J$ at the frustration-free point along the deformation ($h/J=1/3\approx 0.33$) and the transition point in the absence of deformation ($h/J \approx 0.293$) are very close, with the latter being slightly smaller.

\section{Discussion}\label{sec:discuss}

In this work, we proposed symmetry-preserving deformations of classical LDPC codes and quantum CSS codes in transverse field, and identified frustration-free points along these deformations with interesting features: (i) the ground states include a trivial product state and the code space of the original code, and (ii) they remain gapped in the thermodynamic limit. In other words, the frustration-free point realises a coexistence of the trivial phase and the non-trivial gapped phase described by the underlying code. Furthermore, if the original code in transverse field has a non-invertible duality symmetry, then the deformation preserves it, and it is spontaneously broken at the frustration-free point.

\paragraph{Phase diagram.} We analysed several familiar examples of classical LDPC codes and quantum CSS codes on Euclidean lattices in diverse dimensions. In all these examples, the phase diagram is expected to look like one of the two possibilities on the left in Figure \ref{fig:phase-diag}. When there is a non-invertible duality symmetry, the entire transition line is along the self-dual line $h=J$. For instance, in the 1+1d and 2+1d Ising models, and the 2+1d Toric Code, the transition point (TP) at $\lambda = 0$ is second-order, so the phase diagram is the one in the middle. For the other examples, TP is first-order, so the phase diagram is the one on the left.

\begin{figure}
\centering
\hfill \raisebox{-0.5\height}{\includegraphics[scale=0.23]{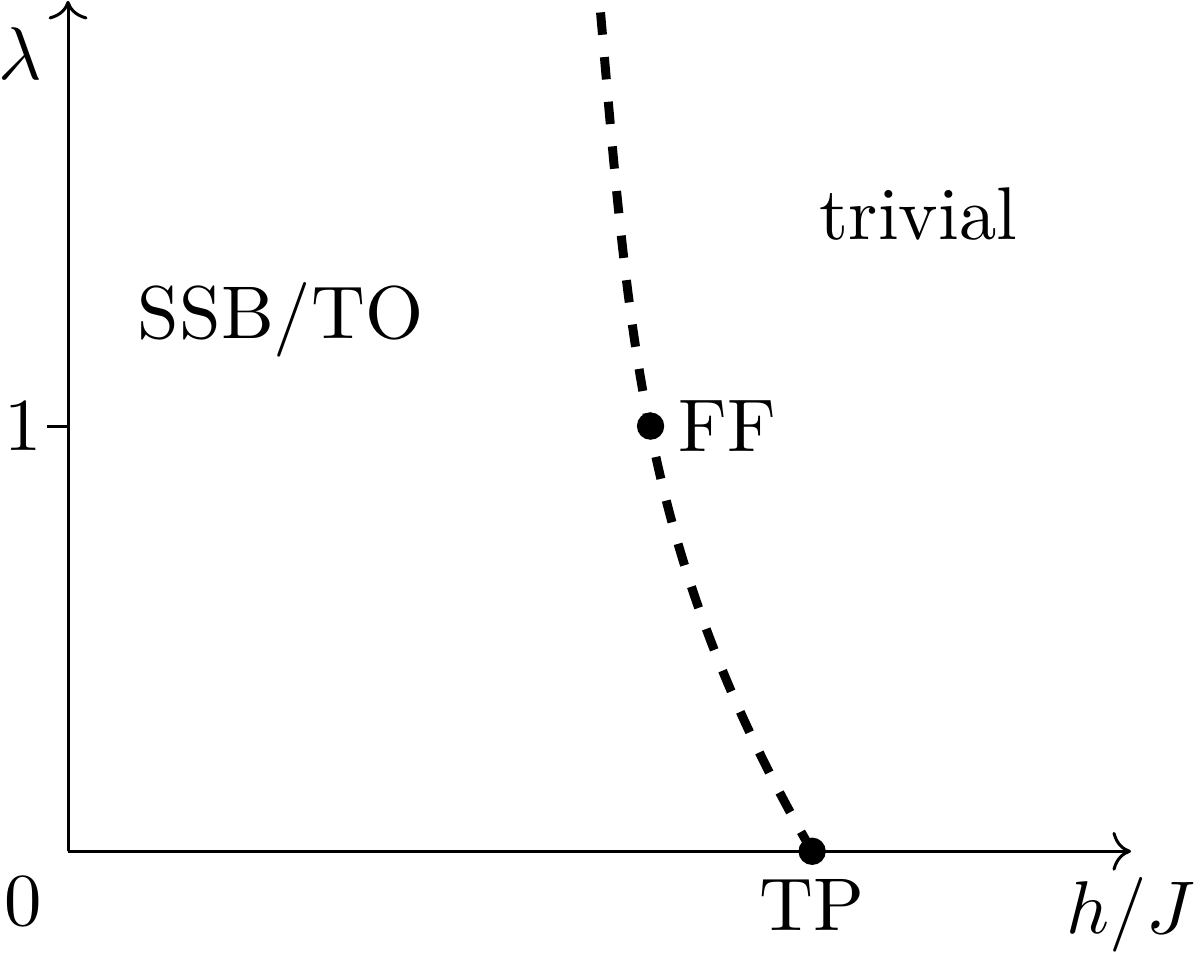}} \hfill \raisebox{-0.5\height}{\includegraphics[scale=0.23]{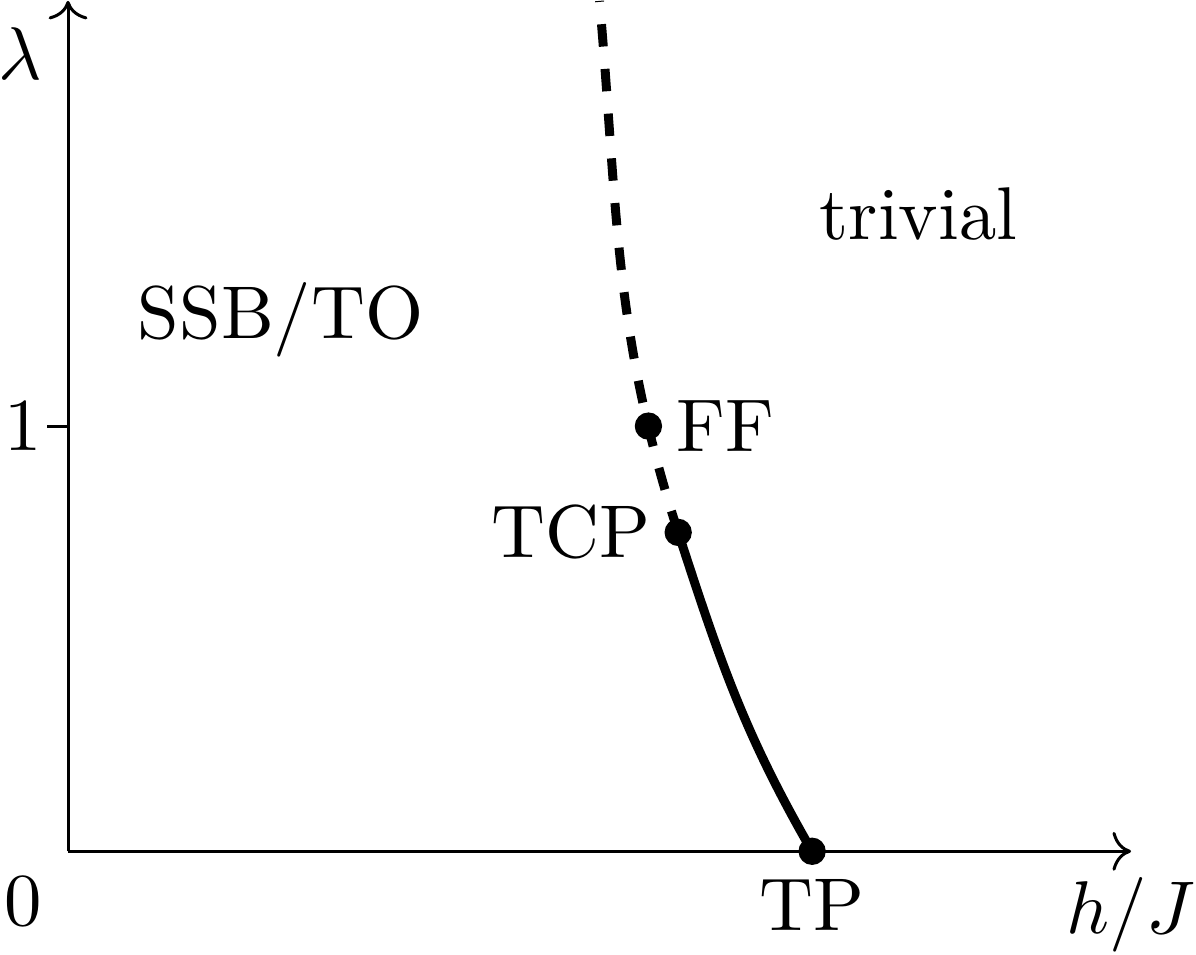}} \hfill \raisebox{-0.5\height}{\includegraphics[scale=0.23]{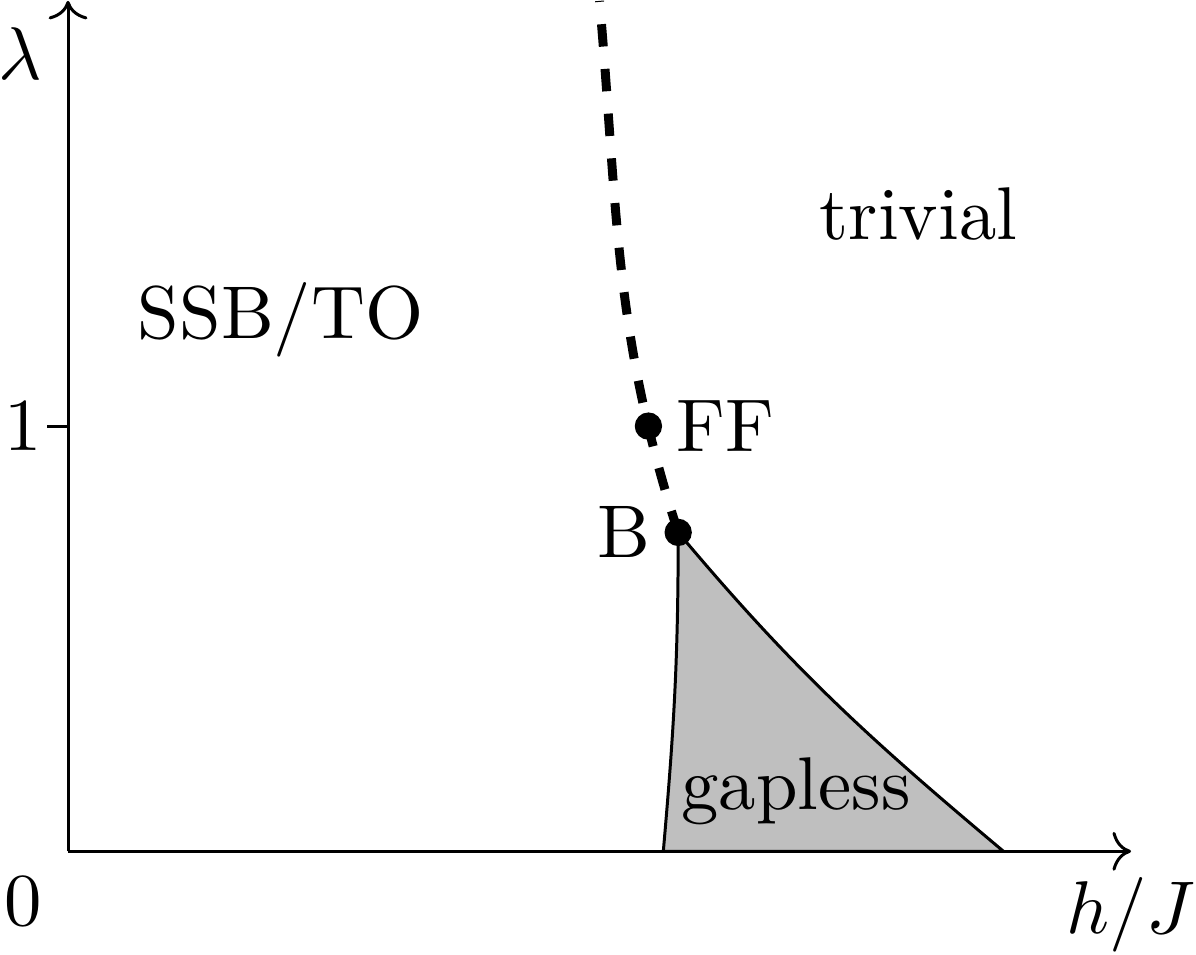}} \hfill ~
\caption{The possible phase diagram of the deformed model on a Euclidean lattice. The dashed (solid) line represents a first (second/higher) order transition between the trivial phase and the spontaneously symmetry broken (SSB) or topologically ordered (TO) phase (the latter includes fracton order). TP denotes the transition point between the two phases at $\lambda = 0$, FF denotes the frustration-free point at $\lambda = 1$, TCP denotes a tricritical point, and B denotes a bifurcation point. The transition line shown here is just a representative, and depending on the model, it could bend the other way. In particular, when there is a non-invertible duality symmetry, the entire transition line is vertical along the self-dual line $h=J$.}\label{fig:phase-diag}
\end{figure}

The phase diagram along the transition line is also quite interesting. Recall that the frustration-free point (FF) at $\lambda = 1$ always realises a first-order transition. Moreover, it was suggested in \cite{OBrien:2017wmx} that the frustration-free point of \eqref{obrien-fendley-deform} could be the basin of the gapped phase along the transition line. Let us assume that this is true more generally for any deformed model \eqref{graph-deformedH} on a Euclidean lattice. Now, if TP is also first-order, then it is reasonable to expect that the entire transition line is gapped and FF is the basin for this gapped phase, as in the phase diagram on the left. On the other hand, if TP is second-/higher-order, then there must be a \emph{tricritical point} (TCP) along the transition line between TP and FF, as in the middle phase diagram. In fact, TCP must be an unstable fixed point of RG, which flows to a gapped phase in the direction of FF and a gapless phase in the direction of TP. For instance, in the 1+1d Ising model, TP is second-order described by the 2d Ising CFT and TCP is described by the 2d tricritical Ising CFT \cite{OBrien:2017wmx}, and indeed, the latter has a relevant perturbation which flows to the 2d Ising CFT or a gapped phase with three vacua depending on its sign \cite{Huse:1984mn}. A more interesting example is the 2+1d Ising model, where TP is second-order described by the 3d Ising CFT, but the status of TCP is still open. We leave the detailed investigation of such novel tricritical points to the future.

\paragraph{Generalising to $\mathbb Z_N$ qudits.} The full phase diagram gets even richer when the qubits are replaced by $\mathbb Z_N$ qudits. All our results go through straightforwardly in this case with minor cosmetic changes: (i) the matrix $\mathfrak h$ now has entries in $\{0,1,\ldots,N-1\} \cong \mathbb Z_N$, which can be interpreted as adding weights to the edges, but the underlying bipartite graph remains the same, and (ii) $P_v$ and $Q_{\hat v}$ should be replaced by $P_v + P_v^\dagger$ and $Q_v + Q_v^\dagger$ to ensure that the Hamiltonians \eqref{graph-deformedH'} and \eqref{css-deformedH'} are Hermitian.\footnote{One can also consider other replacements, such as Potts-like terms instead of clock-like terms.} Now, for the deformed $\mathbb Z_N$ model on a Euclidean lattice, when $N$ is large enough, there is a possibility of a gapless window between the two gapped phases at $\lambda = 0$, as shown in the right phase diagram. But FF at $\lambda = 1$ is still a direct first-order transition between the two gapped phases, so the gapless window must shrink along the deformation to a \emph{bifurcation point} (B) at some $0 < \lambda < 1$. Understanding the field theory around this point would be another interesting future direction.

Besides understanding the phase diagram, there are several potential future directions:
\begin{enumerate}
\item An important question that we have not addressed in this work is the stability of the deformed model \eqref{graph-deformedH'}. It is known that not all LDPC codes correspond to stable gapped phases, and the ones that are stable require additional restrictions---e.g., the code distance must grow with the number of qubits \cite{Rakovszky:2023fng,Rakovszky:2024iks}. One can ask if the deformed model based on such an LDPC code is similarly stable under local (possibly symmetry-preserving) perturbations.

\item Can we extend the symmetry-preserving deformation to models based on non-abelian groups, or even fusion categories? The expectation is to obtain a frustration-free model realising the coexistence of trivial and arbitrary SSB/SPT/TO phases?\footnote{Here, an SPT phase refers to a symmetry-protected topological phase.} What about realising coexistence of two arbitrary gapped phases? Of course, this should not always be possible---e.g., the two gapped phases could have different anomalies in the IR. But when the anomalies match, intuitively, this should be possible.

\item In the examples where there is no non-invertible duality symmetry, is there a different symmetry (perhaps non-invertible) that exchanges the two phases? One could potentially construct such a symmetry from a sequential quantum circuit that maps between the two gapped phases \cite{Chen:2023qst}.

\item It would be interesting to explore the applications of our coarse-graining/blocking procedure on arbitrary (bounded-degree, bipartite) graphs to other local, frustration-free, non-commuting projector Hamiltonians. For instance, one could use it to obtain local gap thresholds of frustration-free Hamiltonians on non-Euclidean lattices, similar to the known thresholds on Euclidean lattices \cite{Gosset:2016,Lemm:2019pxa,Anshu:2019mhi}.
\end{enumerate}

\section*{Acknowledgements}
We thank Anurag Anshu, Michael Kastoryano, Angelo Lucia, Zhu-Xi Luo, Abhinav Prem, Marvin Qi, Shu-Heng Shao, and Nathanan Tantivasadakarn for stimulating discussions. We are grateful to Paul Fendley, Angelo Lucia, and Nathanan Tantivasadakarn for thoughtful comments on the draft. We are indebted to Dömötör Pálvölgyi for help with the proof of existence of a good cover of a hypergraph. This work was supported by the Simons Collaboration on Global Categorical Symmetries.

\appendix

\section{Proof of Claim \ref{clm:eq-wt}}\label{app:eq-wt}
In this appendix, we prove Claim \ref{clm:eq-wt}. Recall that the constraints \eqref{psi-constraint} can be phrased as follows. Given an X-state $|\sigma\> \ne |\Plus\>$, a vertex $v$ with $\sigma_v = -$ and a vertex $\hat v$ such that $d(v,\hat v) = 3$, let $|\sigma'\> = F_{\hat v} |\sigma\>$. We refer to such flips as \emph{special flips}. Then, by \eqref{psi-constraint}, we have $\psi_{\sigma'} = \psi_\sigma$. That is, X-states related by special flips have equal weights in the decomposition \eqref{psi-decomp}.

Our goal is to show that $\psi_{\sigma'} = \psi_\sigma$ whenever $|\sigma\>$ and $|\sigma'\>$, neither of which is $|\Plus\>$, carry the same charges under $\mathbb Z_2^\nu$. Recall that, from the discussion below \eqref{eta-product-eigenstates}, there is a subset $\widehat U \subseteq \widehat V$ such that $|\sigma'\> = \prod_{\hat v \in \widehat U} F_{\hat v} |\sigma\>$. So all we need to show is that we can ``simulate'' the flip $F_{\hat v}$ for all $\hat v\in \widehat U$ using only the special flips.

Given an X-state $|\sigma\> \ne |\Plus\>$, say we want to simulate the flip $F_{\hat v}$. There are three scenarios:
\begin{enumerate}
\item There is a vertex $v \not\sim \hat v$ such that $\sigma_v = -$. Since $\mathcal G$ is connected, there is a $v\not\sim \hat v$ such that $\sigma_v = -$ \emph{and} there is a path from $v$ to $\hat v$,
\ie\nonumber
\begin{tikzcd}[column sep = 0.2cm]
\underset{\color{blue}\phantom{v=w}-}{v =w_1} \ar[r,-] &|[yshift=1cm]| \hat w_1 \ar[r,-] & \underset{\color{blue}+}{w_2} \ar[r,-] &|[yshift=1cm]| \hat w_2 \ar[r,-] & \underset{\color{blue}+}{w_3} \ar[r,-] &|[yshift=1cm]| \hat w_3 \ar[r,-] &[0.2cm]|[yshift=0.5cm]|  \cdots  \ar[r,-] &[0.2cm] \underset{\color{blue}+}{w_{\ell-1}} \ar[r,-] &|[yshift=1cm]| \hat w_{\ell-1} \ar[r,-] & w_\ell \ar[r,-] &|[yshift=1cm]| \hat w_\ell = \hat v
\end{tikzcd}
\fe
with $\ell\ge 2$ and $\sigma_{w_i} = +$ for $2\le i \le \ell-1$, i.e., all internal vertices of the path have only $+$ on them.\footnote{If an internal vertex along this path has a $-$, then redefine $v$ as the internal vertex with $-$ that is closest to $\hat v$ along this path.} Since $d(w_1,\hat w_2)=3$ and $\sigma_{w_1}=-$, we can apply the special flip $F_{\hat w_2}$, so that $w_2$ and $w_3$ have $-$ after this flip. Since $d(w_3,\hat w_4)=3$ and $w_3$ has $-$, we can apply the special flip $F_{\hat w_4}$, so that $w_4$ and $w_5$ have $-$ after this flip. We repeat this process until we apply either the special flip $F_{\hat w_{\ell-1}}$ if $\ell$ is odd or the special flip $F_{\hat w_\ell}$ if $\ell$ is even. When $\ell$ is odd, we apply one more special flip $F_{\hat w_\ell}$, which is allowed because $d(w_{\ell-1},\hat w_\ell)=3$ and $w_{\ell-1}$ has $-$. Now, we ``undo'' all the special flips, except for the last one, by applying them in the reverse order. That this is possible is easy to check. In the end, we have simulated the flip $F_{\hat w_\ell} = F_{\hat v}$, which is exactly what we wanted to show.

\item Every $v \not \sim \hat v$ has $\sigma_v =+$ (i.e., we are not in scenario 1), but there are two vertices $v_1,v_2$ adjacent to $\hat v$ with $\sigma_{v_1} = -$ and $\sigma_{v_2}=+$. Since $|V| > 2D\widehat D$ and $\mathcal G$ is connected by assumption, there are vertices $w\in V$ and $\hat w\in \widehat V$ that satisfy $w\sim \hat w$, $\hat w\not \sim v_{1,2}$, and $w\not\sim \hat v$.\footnote{To see this, consider the union of balls $B(v_1,2)\cup B(v_2,2)$. It contains at most $2D\widehat D$ vertices of $V$. Since $|V| > 2D\widehat D$ by assumption, there is a vertex $w \in V$ such that $d(w,v_{1,2})>2$. It follows that $w \not\sim \hat v$. Since $\mathcal G$ is connected, there must be a $\hat w \in \widehat V$ such that $\hat w \sim w$. Moreover, $\hat w\not\sim v_{1,2}$ because $d(w,v_{1,2}) > 2$.} By the hypothesis of this scenario, $w\not\sim \hat v$ implies that $\sigma_w=+$. Now, $v_1$ and $\hat w$ satisfy scenario 1, so we can simulate the flip $F_{\hat w}$, so that $w$ has $-$ after this flip. Then, $w$ and $\hat v$ satisfy scenario 1, so we can simulate the flip $F_{\hat v}$, so that $v_1$ has $+$ and $v_2$ has $-$ after this flip. Finally, $v_2$ and $\hat w$ satisfy scenario 1, so we can simulate the flip $F_{\hat w}$ again, undoing the first flip. In the end, we have simulated the flip $F_{\hat v}$.

\item Every $v\not\sim \hat v$ has $\sigma_v = +$ and every $v\sim \hat v$ has $\sigma_v = -$ (the latter cannot be $+$ because $|\sigma\> \ne |\Plus\>$ by hypothesis). That is, we are not in either of the two scenarios above. We cannot simulate the flip $F_{\hat v}$ in this scenario because that would result in $|\Plus\>$, but there are no constraints involving $\psi_{\Plus}$. Note that, in this scenario, $|\sigma\> = F_{\hat v}|\Plus\>$, which is invariant under the $\mathbb Z_2^\nu$ symmetry.
\end{enumerate}

We are now ready to prove Claim \ref{clm:eq-wt}. There are two cases to consider.
\begin{itemize}
\item \underline{$|\sigma\>$ is not invariant under $\mathbb Z_2^\nu$}: In this case, we are never in scenario 3. Hence, we can simulate the flips $F_{\hat v}$ for all $\hat v \in \widehat U$ in any order.

\item \underline{$|\sigma\>$ is invariant under $\mathbb Z_2^\nu$}: Say we have already simulated the flips $F_{\hat w}$ for all $\hat w \in \widehat W \subseteq \widehat U$. After these flips, the state is $|\sigma^{\widehat W}\> := \prod_{\hat w \in \widehat W} F_{\hat w} |\sigma\>$. Now pick a new vertex $\hat v \in \widehat U \sm \widehat W$. If $|\sigma^{\widehat W}\> \ne F_{\hat v} |\Plus\>$, then we are not in scenario 3, so we can simulate the flip $F_{\hat v}$, add $\hat v$ to $\widehat W$, and proceed to the next flip. On the other hand, if $|\sigma^{\widehat W}\> = F_{\hat v} |\Plus\>$, then we pick a different vertex $\hat v \ne \hat v' \in \widehat U \sm \widehat W$ and proceed. We are guaranteed that $|\sigma^{\widehat W}\> \ne F_{\hat v'} |\Plus\>$, so we are not in scenario 3. The only potential issue is when $\hat v$ is the last vertex in $\widehat U$, i.e., $\widehat W = \widehat U \sm \{\hat v\}$, and we do not have the option to pick another vertex. But since $|\sigma'\> \ne |\Plus\>$ by the hypothesis of Claim \ref{clm:eq-wt}, it cannot be the case that $|\sigma^{\widehat U \sm \{\hat v\}}\> = F_{\hat v} |\Plus\>$, so we are guaranteed to not be in scenario 3. In the end, we have simulated the flip $F_{\hat v}$ for all $\hat v \in \widehat U$.
\end{itemize}
This proves the claim.

\section{Coarse-graining a hypergraph}\label{app:coarse-grain}

In this appendix, we define the interaction hypergraph of a local Hamiltonian and give a procedure to coarse-grain it so that the delicate balance discussed in Section \ref{sec:mart-meth} is achieved. But first, we mention some useful notions of hypergraphs.

\subsection{Basic notions of hypergraphs}
Let $G$ be a hypergraph with vertex-set $V$ and hyperedge-set $E \subseteq 2^{V} \sm \{\emptyset\}$, where $2^V$ is the power set of $V$. The \emph{degree} of a vertex is the number of hyperedges it is contained in and the \emph{order} of a hyperedge is the number of vertices it contains. For instance, if $G$ is an ordinary graph, then every hyperedge is of order $2$. Let $\mdeg(G) := \max_{v\in V} \deg(v)$ be the maximum of degrees of all vertices and $\mord(G) := \max_{e\in E} |e|$ be the maximum of orders of all hyperedges, respectively. The \emph{incidence matrix} $\mathfrak i$ of $G$ is a $|E| \times |V|$ matrix given by $\mathfrak i_{e,v} = 1$ if $v\in e$ and $0$ otherwise.

Two hyperedges are \emph{adjacent} if they share a vertex and two vertices are \emph{adjacent} if there is a hyperedge that contains them. A \emph{path} between two vertices $u$ and $v$ is a sequence of distinct vertices $u=v_1,v_2,\ldots,v_n=v$ such that $v_i$ and $v_{i+1}$ are adjacent for any $1\le i < n$. We say $G$ is \emph{connected} if there is a path between any two vertices. We define $\mathscr G(\Delta,k)$ as the collection of all finite connected hypergraphs with $\mdeg(G) \le \Delta$ and $\mord(G) \le k$, where $\Delta$ and $k$ are positive integers.

We say $G'$ is a \emph{sub-hypergraph} of $G$ if $V(G') \subseteq V(G)$, $E(G') \subseteq E(G)$, and $\bigcup_{e\in E(G')} e \subseteq V(G')$. We say $G'$ is \emph{induced} if $e\in E(G)$ and $e\subseteq V(G')$ implies $e\in E(G')$. We say $G'$ is a \emph{connected component} of $G$ if it is induced, connected, and no hyperedge in $E(G)\sm E(G')$ intersects $G'$. In particular, $G$ is connected if and only if it has only one connected component.

The \emph{intersection} of two sub-hypergraphs $G_1$ and $G_2$ is a sub-hypergraph, denoted as $G_1\cap G_2$, whose vertex-set is $V(G_1) \cap V(G_2)$ and hyperedge-set is $E(G_1) \cap E(G_2)$. It follows that, if $G_1$ and $G_2$ are both induced, then so is $G_1 \cap G_2$. However, even if $G_1$ and $G_2$ are connected, $G_1 \cap G_2$ need not be connected. The \emph{union} of two sub-hypergraphs is defined similarly. Note that, unlike the intersection, $G_1 \cup G_2$ is connected if $G_1$ and $G_2$ are connected and $G_1 \cap G_2 \ne \emptyset$, but it need not be induced even if $G_1$ and $G_2$ are induced.

A \emph{connected partition} of $G$ is a partition of $V(G)$ such that each part induces a connected sub-hypergraph. Two parts of a partition are said to be \emph{adjacent} if there is a hyperedge that intersects both the parts nontrivially (i.e., if there is a hyperedge that contains a vertex from both parts).

The \emph{incidence graph} of $G$ is the bipartite graph $BG$ with vertex set $BV:=E\sqcup V$ and edge-set $BE := \{\{e,v\}:e \in E,v\in V,v\in e\}$. In other words, $BG$ is the bipartite graph whose biadjacency matrix is the incidence matrix of $G$. This gives a one-one correspondence between hypergraphs and bipartite graphs. For instance, $G$ is connected if and only if $BG$ is connected. Also, $\mdeg(G) = \mrdeg(BG)$ and $\mord(G) = \mldeg(BG)$. A sub-hypergraph $G'$ of $G$ corresponds to a left-closed subgraph $BG'$ of $BG$, and as the notation suggests, $BG'$ is the incidence graph of $G'$. A non-left-closed subgraph of $BG$ does not correspond to a natural sub-hypergraph of $G$.

\subsection{Interaction hypergraph of a Hamiltonian}\label{app:int-hypergraph}

A local Hamiltonian on a tensor product Hilbert space has the following structure. Let $v\in V$ label the local Hilbert space $\mathcal H_v$. The local Hamiltonian can be written as $H= \sum_{e\in E} H_e$, where $e$ labels the local interactions. We can associate $e$ with the subset of qubits that participate in the interaction labelled by $e$, i.e., $e\subseteq V$. This yields a hypergraph with vertex-set $V$ and hyperedge-set $E$, which is referred to as the \emph{interaction hypergraph} of the Hamiltonian $H$. The corresponding incidence graph is a bipartite graph, which is commonly referred to as the \emph{interaction (bipartite) graph} of the Hamiltonian $H$.

A natural notion of a ``local, short-ranged Hamiltonian'' is when each $H_e$ involves at most $k$ qubits, and each qubit participates in at most $\Delta$ interaction terms. The corresponding interaction hypergraph belongs to $\mathscr G(\Delta,k)$.

Let us consider two examples discussed in the main text.
\begin{itemize}
\item \textbf{Example 1:} When $h_{\hat v} = 0$ for all $\hat v \in \widehat V$, the interaction (bipartite) graph of the Hamiltonian \eqref{H-gTFIM} of the classical LDPC code in transverse field is precisely the bipartite graph $\mathcal G$ (i.e., the Tanner graph) that it is based on. Therefore, in this case, we can choose $\Delta = D$ and $k = \widehat D$.

\item \textbf{Example 2:} Consider the Hamiltonian \eqref{graph-deformedH'}, which is a deformation of the Hamiltonian in Example 1. Each term in this Hamiltonian involves the qubit on $v$ and the qubits on the neighbours of $\hat v$. Therefore, each interaction term involves at most $k = 1+\widehat D$ qubits. On the other hand, each qubit participates in an interaction term either as $v$ via $P_v$ or as a neighbour of $\hat v$ via $Q_{\hat v}$. For each $v$, $P_v$ appears in at most $D^2 \widehat D$ terms because there are at most that many $\hat v$'s at distance $3$ from $v$. So each $v$ appears at most $D^2 \widehat D$ times via $P_v$. Similarly, for each $\hat v$, $Q_{\hat v}$ appears in at most $D \widehat D^2$ terms. Since each $v$ has at most $D$ neighbours, it appears at most $D^2 \widehat D^2$ times via $Q_{\hat v}$. In total, each qubit participates in at most $\Delta = D^2 \widehat D(1+\widehat D)$ interaction terms. To summarise, the interaction hypergraph of the Hamiltonian \eqref{graph-deformedH'} belongs to $\mathscr G(\Delta,k)$, where $\Delta = D^2 \widehat D(1+\widehat D)$ and $k= 1+\widehat D$. Since $D$ and $\widehat D$ are fixed, $\Delta$ and $k$ are also fixed in the thermodynamic limit.

\end{itemize}

\subsection{Good cover of hypergraph}

Given positive integers $n,\Delta,k,m$, and $N\ge n$, and a finite connected hypergraph $G \in \mathscr G(\Delta,k)$ with $|V(G)| \ge n$, we define a \emph{good cover} of $G$ with parameters $m,N$ as a finite collection $\{G_i\}$ of sub-hypergraphs of $G$ such that:
\begin{enumerate}
\item every $G_i$ is a connected induced sub-hypergraph,

\item every $e\in E$ is contained in at least one $G_i$ and at most $m$ distinct $G_i$'s,

\item for any $i$, $n \le |V(G_i)| \le N$, and

\item for any $i\ne j$, if $G_i$ and $G_j$ intersect, then $G_i \cap G_j$ has at least one connected component with at least $n$ vertices.

\end{enumerate}
Of course, for any given hypergraph, one can always find $m$ and $N$ such that the above conditions are all satisfied (provided $|V(G)| \ge n$). The question is if one can choose $m,N$ \emph{independent} of the hypergraph. That is, are there $m,N$ such that for all sufficiently large $n$, \emph{every} hypergraph with at least $n$ vertices, degree at most $\Delta$, and order at most $k$ has a good cover with parameters $m,N$? The answer is yes!

\begin{theorem}\label{thm:goodcover}
For any $n \ge k$, every hypergraph $G\in\mathscr G(\Delta,k)$ with $|V(G)|\ge n$ has a good cover with parameters
\ie
m=m_{\Delta,k}(n) := 2^k \Delta^{2k} (k!)^2 n^{k(k+1)}~,\qquad N=N_{\Delta,k}(n) := 2^k \Delta^k\,k!\, n^{(k+1)(k+2)/2}~.
\fe
\end{theorem}

Before proving this theorem, let us comment on its implications for our model. In Section \ref{sec:coarse-grain-bi-graph}, we claimed that there is a good cover of a bipartite graph with four properties. This follows immediately by applying the above theorem to the interaction hypergraph of the Hamiltonian \eqref{graph-deformedH'} discussed in Example 2 of Appendix \ref{app:int-hypergraph}. We omit the details and proceed to prove Theorem \ref{thm:goodcover}.

\begin{proof}[Proof of Theorem \ref{thm:goodcover}]
Any connected hypergraph of maximum order at most $k=1$ has at most one vertex, so the statement is trivially true when $k=1$. For $k>1$, we proceed by induction with the base case $k=1$ being trivially true.

We define a \emph{good part} as a subset of $V(G)$ with at least $n$ vertices and at most $n+k$ vertices such that the induced sub-hypergraph is connected. One can construct a good part by picking a hyperedge, then another hyperedge adjacent to it, then another hyperedge adjacent to one of the last two, and so on. We stop when we have at least $n$ vertices. Since the maximum order of $G$ is at most $k$, this part will have at most $n+k$ vertices, so it is a good part---call it $B_1$. Since $|V(G)|\ge n$, there is always at least one good part.

If a connected component of the remaining vertices $V(G)\sm B_1$ has fewer than $n$ vertices, call it a \emph{bad part}---call it $C_1$; else, if it has at least $n$ vertices, then repeat the above procedure to construct a new good part. Repeat these steps as long as possible. In the end, we are left with good parts $\mathcal B = \{B_i\}$ and bad parts $\mathcal C = \{C_a\}$, and together, they form a connected partition of $V(G)$.

Given this partition of vertices, the hyperedges of $G$ can be divided into four types: I-B are those that are contained within a single good part, I-C are those that are contained within a single bad part, II are those that intersect multiple good parts but no bad part, and III are those that intersect at least one good part and at least one bad part. Note that there is no hyperedge that intersects multiple bad parts but no good part because otherwise, those bad parts form a larger connected component, which can either be a bigger bad part (reducing the number of bad parts by at least 1 while keeping the number of good parts the same), or yield a new good part (increasing the number of good parts by at least 1). Therefore, the above four types exhausts all hyperedges.

Consider a new hypergraph $G'$ whose vertex set is the collection of bad parts, i.e., $V(G') = \mathcal C$. The hyperedges $E(G')$ are constructed as follows: if there is a hyperedge $e\in E(G)$ that intersects more than one bad part, then add a hyperedge containing those bad parts to $E(G')$. For example, if $e$ intersects $C_1,C_2,C_5$ and no other $C_a$, then add the hyperedge $\{C_1,C_2,C_5\}$ to $E(G')$. Moreover, $e$ must intersect a good part as explained in the last paragraph, so each hyperedge in $E(G')$ contains at most $k-1$ vertices. Also, each bad part contains at most $n$ vertices of $G$, each with degree at most $\Delta$, so the maximum degree of $G'$ is at most $\Delta n$. Therefore, $G'\in\mathscr G(\Delta n,k-1)$.

We now construct a new collection $\mathcal C'$. In general, $G'$ is not connected, so pick a connected component $H'$ of $G'$.\footnote{$H'$ here should not be confused with the Hamiltonian.} If $H'$ has fewer than $n$ vertices, then add it to $\mathcal C'$. On the other hand, if $H'$ has at least $n$ vertices, then by the induction hypothesis, there is a good cover of $H'$ with parameters $m_{\Delta n,k-1}(n)$ and $N_{\Delta n,k-1}(n)$; add the elements of this good cover to $\mathcal C'$. We shall refer to the elements of $\mathcal C'$ as \emph{bad patches} for convenience.\footnote{They should not be referred to as ``parts'' because they do not form a partition of bad vertices in the original hypergraph $G$.}

Let us finally construct a good cover of $G$ using the collections $\mathcal B$ and $\mathcal C'$. For each good part $B\in \mathcal B$, define $G_B$ as the induced sub-hypergraph on the vertices in the union of $B$ and all the good parts that are adjacent to $B$. Similarly, for each bad patch $C'\in \mathcal C'$, define $G_{C'}$ as the induced sub-hypergraph on the vertices in the union of $C'$ and all the good parts adjacent to $C'$. We claim that $\{G_B : B\in \mathcal B\} \cup \{G_{C'} : C'\in \mathcal C'\}$ is a good cover of $G$.

\begin{enumerate}
\item Each $G_B$ and each $G_{C'}$ is a connected induced sub-hypergraph of $G$ by construction. Therefore, condition 1 of good cover is satisfied.

\item Any type I-B or type II hyperedge is contained in some $G_B$. Any type I-C or type III hyperedge is contained in some $G_{C'}$.

Each good part has at most $n+k$ vertices, so it has at most $\Delta (k-1) (n+k)$ adjacent parts. On the other hand, each vertex of $G'$ has degree at most $\Delta n$ and each hyperedge of $G'$ is contained in at most $m_{\Delta n,k-1}(n)$ bad patches by the definition of good cover, so each bad part is covered by at most $\Delta n \times m_{\Delta n,k-1}(n)$ bad patches.
\begin{itemize}
\item Any type I-B edge contained in a good part $B$ is covered by (i) $G_B$, (ii) all $G_{\bar B}$, where $\bar B$ is adjacent to $B$, and (iii) all $G_{C'}$, where $C'$ is adjacent to $B$. This gives
\ie
m^{\text{I-B}}_{\Delta,k}(n) &\le 1+ \Delta (k-1) (n+k) \times \Delta n \times m_{\Delta n,k-1}(n)
\\
&\le 2\Delta^2 k n^2 m_{\Delta n,k-1}(n)~,
\fe
where we used $n \ge k$.

\item Any type II edge that intersects good parts, say, $B_1,\ldots,B_\ell$ (with $\ell\le k$) is covered by (i) $G_{B_1},\ldots, G_{B_\ell}$, (ii) some $G_{\bar B}$, where $\bar B$ is adjacent to $B_j$ for some $1\le j\le \ell$, and (iii) some $G_{C'}$, where $C'$ is adjacent to $B_j$ for some $1\le j\le \ell$. This gives
\ie
m^{\text{II}}_{\Delta,k}(n) \le \ell \times m^{\text{I-B}}_{\Delta,k}(n) \le 2\Delta^2 k^2 n^2 m_{\Delta n,k-1}(n)~.
\fe

\item Any type I-C edge contained in single a bad part $C$ is covered by all $G_{C'}$, where $C'$ is a bad patch that contains $C$. This gives
\ie
m^{\text{I-C}}_{\Delta,k}(n) \le 
\Delta n m_{\Delta n,k-1}(n)~.
\fe

\item Any type III edge corresponds to a unique hyperedge of $G'$, so it is covered by $m_{\Delta n,k-1}(n)$ bad patches. This gives
\ie
m^{\text{III}}_{\Delta,k}(n) \le m_{\Delta n,k-1}(n)~.
\fe
\end{itemize}
Therefore, we can satisfy condition 2 of good cover by choosing
\ie
m_{\Delta,k}(n) = 2\Delta^2 k^2 n^2 m_{\Delta n,k-1}(n)~.
\fe
Solving the recursion relation gives
\ie
m_{\Delta,k}(n) = 2^k \Delta^{2k} (k!)^2 n^{k(k+1)}~,
\fe
which grows polynomially in $n$ as required.

\item Each good part has at least $n$ vertices and at most $n+k$ vertices. Each bad patch has at least one and at most $N_{\Delta n,k-1}(n)$ bad parts, which translates to at least one and at most $n \times N_{\Delta n,k-1}(n)$ vertices of $G$.
\begin{itemize}
\item Since $G_B$ contains the good part $B$, it has at least $n$ vertices. And it contains at most $1+ \Delta (k-1) (n+k)$ good parts, each with at most $n+k$ vertices, so it has at most
\ie
[1+ \Delta (k-1) (n+k)]\times (n+k) \le 4\Delta k n^2
\fe
vertices, where we used $n\ge k$.

\item Since each bad patch $C'$ is adjacent to at least one good part, $G_{C'}$ contains at least $n$ vertices. Moreover, each bad patch $C'$ contains at most $N_{\Delta n,k-1}(n)$ bad parts, and each bad part is adjacent to at most $\Delta (k-1) n$ good parts, so $G_{C'}$ contains at most 
\ie
[n+\Delta (k-1) n \times (n+k)] N_{\Delta n,k-1}(n) \le 2\Delta k n^2 N_{\Delta n,k-1}(n)
\fe
vertices.
\end{itemize}
We can satisfy condition 3 of good cover by choosing
\ie
N_{\Delta,k}(n) = 2\Delta k n^2 N_{\Delta n,k-1}(n)~.
\fe
Solving the recursion relation gives
\ie
N_{\Delta,k}(n) = 2^k \Delta^k\,k!\, n^{(k+1)(k+2)/2}~,
\fe
which also grows polynomially in $n$ as required.

\item It is easy to check that if any two elements of $\{G_B : B\in \mathcal B\} \cup \{G_{C'} : C'\in \mathcal C'\}$ intersect, they do so in a good part, which is connected and has at least $n$ vertices by construction. Therefore, condition 4 is also satisfied.
\end{enumerate}
This completes the proof of Theorem \ref{thm:goodcover}.
\end{proof}

\section{Upper bound on the martingale function}\label{sec:graph-delta}

In this appendix, we derive the upper bound \eqref{up-bnd-delta} on the martingale function for the deformed model given by the Hamiltonian \eqref{graph-deformedH'}.

Let $A$ be a connected induced left-closed subgraph of $\mathcal G$ on the vertices $\widehat V_A \sqcup V_A$ with $|V_A| > 2D \widehat D$. Let $\mathfrak h_A$ be the $|\widehat V_A| \times |V_A|$ biadjacency matrix of $A$, i.e., $\mathfrak h_A$ is the restriction of $\mathfrak h$ to the rows and columns associated with $\widehat V_A$ and $V_A$, respectively. Let $\mathcal B_A$ be a basis of $\ker \mathfrak h_A$ and $\nu_A := |\mathcal B_A| = \dim \ker \mathfrak h_A$.

Let $\mathcal H_A := \bigotimes_{v \in V_A} \mathcal H_v$ be the tensor product Hilbert space of qubits in $A$. For any $a\in\{0,1\}^{|V_A|}$, let us also define the product state $|a\>_A:= \prod_{v\in V_A} X_v^{a_v} |\Zero\>_A$, where $|\Zero\>_A := |0\cdots0\>_A$. In other words, $|a\>_A$ is the eigenstate of $Z_v$ with eigenvalue $a_v$.

Consider the restriction of the Hamiltonian $H'$ \eqref{graph-deformedH'} to $A$,
\ie
H'_A := \sum_{v,\hat v\in A\,:\, d(v,\hat v)=3} P_v Q_{\hat v}~.
\fe
By the assumptions on $A$, using an argument similar to the one in Section \ref{sec:graph-proofofGSD}, we can show that the ground states of $H'_A$ are,
\ie
&|\Plus\>_A := |{+}\cdots{+}\>_A~,\qquad |a\>_A = \eta_{a|A} |\Zero\>_A~,
\fe
where $a\in \ker \mathfrak h_A$, and we defined the $\mathbb Z_2$ symmetry operator
\ie
\eta_{a|A} := \prod_{v\in V_A} X_v^{a_v}~,
\fe
which commutes with the Hamiltonian $H'_A$. There are $1+2^{\nu_A}$ ground states in total. Moreover, while the states $|a\>_A$ are orthogonal to each other, they are not (but almost) orthogonal to $|\Plus\>_A$, i.e.,
\ie
{}_A\<\Plus|a\>_A = {}_A\<\Plus|\Zero\>_A = \frac{1}{2^{|V_A|/2}} > 0~.
\fe
We use $\mathcal H^0_A$ to denote the subspace spanned by $|\Plus\>_A$ and $|a\>_A$, and $\Pi_A$ to denote the orthogonal projection onto $\mathcal H^0_A$, i.e., $\im( \Pi_A) = \mathcal H^0_A$.

Now, consider another connected induced left-closed subgraph $B$ on vertex-set $\widehat V_B \sqcup V_B$, with $|V_B| > 2D \widehat D$, such that $V_{A\cap B} \ne \emptyset$ and $\widehat V_{A\cap B} \ne \emptyset$. Recall that (i) both $A\cup B$ and $A\cap B$ are left-closed, (ii) while $A\cup B$ is connected, $A\cap B$ is not necessarily connected, and (iii) while $A\cap B$ is induced, $A\cup B$ is not necessarily induced. Define $A'$ as the subgraph with vertex-set $\widehat V_{A'} \sqcup V_{A'} := (\widehat V_A \sm \widehat V_B) \sqcup (V_A \sm V_B)$ and no edges, and define $B'$ similarly. It is clear that $A'$ and $B'$ are neither left-closed, nor connected, nor induced.

With the above conventions, the Hilbert space on $A\cup B$ factorises as:
\ie\label{graph-factorisation}
\mathcal H_{A\cup B} = \mathcal H_{A'} \otimes \mathcal H_{A\cap B} \otimes \mathcal H_{B'}~.
\fe
Moreover, since $A$, $B$, $A\cap B$, and $A\cup B$ are all left-closed, the matrix $\mathfrak h_{A\cup B}$ has the following structure:
\ie
\mathfrak h_{A\cup B} = \begin{pNiceMatrix}[first-row,first-col]
 & {\scriptstyle V_{A'}} & {\scriptstyle V_{A\cap B}} & {\scriptstyle V_{B'}} \\
{\scriptstyle \widehat V_{A'}} & * & * & 0 \\
{\scriptstyle \widehat V_{A\cap B}} & 0 & * & 0 \\
{\scriptstyle \widehat V_{B'}} & 0 & * & *\\
\CodeAfter
  \begin{tikzpicture}
  \node [draw=red, thick, rounded corners=2pt, fit = (1-1) (2-2)] {} ;
  \node [draw=green!90!black, thick, rounded corners=2pt, inner sep = 7pt, fit = (2-2) (2-2) ] {} ;
  \node [draw=blue, thick, rounded corners=2pt, fit = (2-2) (3-3) ] {} ;
  \end{tikzpicture}
\end{pNiceMatrix}~,
\fe
where $*$'s represent nontrivial matrices of appropriate dimensions, and the red, green, and blue blocks correspond to $\mathfrak h_A$, $\mathfrak h_{A\cap B}$, and $\mathfrak h_B$, respectively. It follows that, for any $c\in \ker \mathfrak h_{A\cup B}$, we have
\ie\label{graph-factorisation2}
&c_{|A} = c_{|A'} \oplus c_{|A\cap B} \in \ker \mathfrak h_A~,
\\
&c_{|A\cap B} \in \ker \mathfrak h_{A\cap B}~,
\\
&c_{|B} = c_{|A\cap B} \oplus c_{|B'} \in \ker \mathfrak h_B~,
\fe
where $c_{|A}$ denotes the restriction of $c$ to $A$ and so on. Hence, for any $|a\>_A$ and $|b\>_B$, we have
\ie\label{graph-innerproduct}
{}_A\<\Plus|\Plus\>_B &= |\Plus\>_{B'}~{}_{A'}\<\Plus|~,
\\
{}_A\<a|\Plus\>_B &= \frac1{2^{|V_{A\cap B}|/2}} |\Plus\>_{B'}~{}_{A'}\<a_{|A'}|~,
\\
{}_A\<\Plus|b\>_B &= \frac1{2^{|V_{A\cap B}|/2}} |b_{|B'}\>_{B'}~{}_{A'}\<\Plus|~,
\\
{}_A\<a|b\>_B &= {}_{A'}\<a_{|A'}| \left( {}_{A\cap B}\<\Zero| \prod_{v\in V_{A\cap B}} X_v^{a_v + b_v} |\Zero\>_{A\cap B} \right) |b_{|B'}\>_{B'} 
\\
&= \begin{cases}
|b_{|B'}\>_{B'} ~{}_{A'}\<a_{|A'}|~, & \text{if}~~ a_{|A\cap B} = b_{|A\cap B}~,
\\
0~,&\text{otherwise}~.
\end{cases}
\fe
Furthermore, if $a\in \ker \mathfrak h_A$ and $b\in \ker \mathfrak h_B$, and they satisfy $a_{|A\cap B} = b_{|A\cap B}$, then
\ie\label{graph-factorisation3}
a \oplus b_{|B'} = a_{|A'} \oplus b \in \ker \mathfrak h_{A\cup B}~,
\fe
or equivalently, there is a $c\in \mathfrak h_{A\cup B}$ such that $c|_A = a$ and $c|_B = b$.

Finally, let us comment on the projections onto the respective ground state spaces. Note that $ \Pi_A  \Pi_{A\cup B} =  \Pi_{A\cup B}  \Pi_A =  \Pi_{A\cup B}$, and similarly for $ \Pi_B$, because of frustration-freeness. It follows that $ \Pi_A -  \Pi_{A\cup B}$ and $ \Pi_B -  \Pi_{A\cup B}$ are orthogonal projections too. In general, $\Pi_A \wedge  \Pi_B \ne \Pi_{A\cup B}$, or equivalently, $\mathcal H^0_{A\cup B} \ne \mathcal H^0_A \cap \mathcal H^0_B$. While frustration-free implies $\mathcal H^0_{A\cup B} \subseteq \mathcal H^0_A \cap \mathcal H^0_B$, the other direction is not true in general. Fortunately, it is true in our model, i.e., $\Pi_A \wedge  \Pi_B = \Pi_{A\cup B}$, or equivalently, $\mathcal H^0_{A\cup B} = \mathcal H^0_A \cap \mathcal H^0_B$. For a proof, see Appendix \ref{app:proj-wedge-union}.

We are now ready to upper bound the martingale function. Recall its definition:
\ie\label{graph-delta-def}
\delta(A,B) := \Vert  \Pi_A  \Pi_B -  \Pi_A \wedge  \Pi_B \Vert = \Vert  \Pi_A  \Pi_B -  \Pi_{A\cup B} \Vert = \Vert ( \Pi_A - \Pi_{A\cup B})( \Pi_B -  \Pi_{A\cup B}) \Vert~,
\fe
Since the norm of product of orthogonal projections is always at most $1$, we have $0\le \delta(A,B) \le 1$. Below, we will show that
\ie\label{graph-delta-upper-bound}
\delta(A,B) &\le \frac{3\sqrt 2}{\sqrt 2-1} \times 2^{-|V_{A\cap B}|/4\max(D,\widehat D)}~.
\fe
In order to prove this bound, we consider an equivalent definition of $\delta(A,B)$ in terms of states instead of projections:
\ie\label{graph-delta-state-def}
\delta(A,B) &= \sup\{ |\<\psi|\chi\>|: |\psi\>,|\chi\> \in \mathcal H_{A\cup B}\,,\, \<\psi|\psi\> \le 1\,,\,\<\chi|\chi\> \le 1\,,
\\
&\qquad \qquad \qquad \quad( \Pi_A - \Pi_{A\cup B})|\psi\> =|\psi\>\,,\, ( \Pi_B - \Pi_{A\cup B})|\chi\> =|\chi\>\}~.
\fe
Note that the condition $( \Pi_A - \Pi_{A\cup B})|\psi\> =|\psi\>$ means that $|\psi\> \in (\mathcal H^0_A \otimes \mathcal H_{B'}) \cap (\mathcal H^0_{A\cup B})^\perp$, i.e., $|\psi\>$ is a ground state in $A$ and it is orthogonal to the ground states in $A\cup B$,\footnote{Indeed, $ \Pi_A |\psi\> =  \Pi_A ( \Pi_A - \Pi_{A\cup B})|\psi\> = ( \Pi_A - \Pi_{A\cup B})|\psi\> = |\psi\>$, whereas $ \Pi_{A\cup B}|\psi\> =  \Pi_{A\cup B}( \Pi_A - \Pi_{A\cup B})|\psi\> = ( \Pi_{A\cup B} - \Pi_{A\cup B})|\psi\> = 0$.} and similarly for $|\chi\>$. So we write
\ie
&|\psi\> = |\Plus\>_A |\psi_+\>_{B'} + \sum_{a\in \ker \mathfrak h_A} |a\>_A |\psi_a\>_{B'}~,
\\
&|\chi\> = |\chi_+\>_{A'} |\Plus\>_B + \sum_{b\in \ker \mathfrak h_B} |\chi_b\>_{A'} |b\>_B~,
\fe
for some states $|\psi_{+,a}\>_{B'} \in \mathcal H_{B'}$ and $|\chi_{+,b}\>_{A'} \in \mathcal H_{A'}$, so that $|\psi\>$ and $|\chi\>$ are ground states in $A$ and $B$, respectively.

We now restrict them to be orthogonal to the ground states in $A\cup B$. First, we have 
\ie\label{graph-ortho1}
0 = {}_{A\cup B}\<\Plus|\psi\> = {}_{B'}\<\Plus|\psi_+\>_{B'} + \frac{1}{2^{|V_A|/2}} \sum_{a\in \ker \mathfrak h_A}~{}_{B'}\<\Plus|\psi_a\>_{B'}~.
\fe
Next, for any $c\in \ker \mathfrak h_{A\cup B}$, we have  
\ie\label{graph-ortho2}
0 &= {}_{A\cup B}\<c|\psi\> = \frac{1}{2^{|V_A|/2}}~{}_{B'}\<c_{|B'}|\psi_+\>_{B'} + {}_{B'}\<c_{|B'}|\psi_{c_{|A}}\>_{B'}~,
\fe
where we used the fact that $c_{|A} \in \ker \mathfrak h_A$. Similarly, we have
\ie\label{graph-ortho3}
&{}_{A\cup B}\<\Plus|\chi\> = 0 \implies {}_{A'}\<\Plus|\chi_+\>_{A'} + \frac{1}{2^{|V_B|/2}} \sum_{b\in \ker \mathfrak h_B}~{}_{A'}\<\Plus|\chi_b\>_{A'} = 0~,
\\
&{}_{A\cup B}\<c|\chi\> = 0 \implies \frac{1}{2^{|V_B|/2}}~{}_{A'}\<c_{|A'}|\chi_+\>_{A'} + {}_{A'}\<c_{|A'}|\chi_{c_{|B}}\>_{A'} = 0~.
\fe
The constraints on the norms give\footnote{Here, we use the inequality
\ie
{}_{B'}\<\psi_+|\psi_a\>_{B'} + {}_{B'}\<\psi_a|\psi_+\>_{B'} \ge -\frac1\varepsilon~{}_{B'}\<\psi_+|\psi_+\>_{B'} - \varepsilon~{}_{B'}\<\psi_a|\psi_a\>_{B'}~,
\fe
for any $\varepsilon>0$, which follows from $\Vert |\psi_+\>_{B'} + \varepsilon |\psi_a\>_{B'} \Vert \ge 0$. We then substitute $\varepsilon = 2^{\nu_A/2}$.}
\ie\label{graph-norm1}
1&\ge\<\psi|\psi\>
\\
&= {}_{B'}\<\psi_+|\psi_+\>_{B'} + \sum_{a\in \ker \mathfrak h_A}~{}_{B'}\<\psi_a|\psi_a\>_{B'} + \frac{1}{2^{|V_A|/2}} \sum_{a\in \ker \mathfrak h_A} \left( {}_{B'}\<\psi_+|\psi_a\>_{B'} + {}_{B'}\<\psi_a|\psi_+\>_{B'} \right)
\\
&\ge \left( 1-\frac{2^{\nu_A/2}}{2^{|V_A|/2}} \right) \left( {}_{B'}\<\psi_+|\psi_+\>_{B'} + \sum_{a\in \ker \mathfrak h_A}~{}_{B'}\<\psi_a|\psi_a\>_{B'} \right)~,
\fe
and similarly,
\ie\label{graph-norm2}
1 \ge \left( 1-\frac{2^{\nu_B/2}}{2^{|V_B|/2}} \right) \left( {}_{A'}\<\chi_+|\chi_+\>_{A'} + \sum_{b\in \ker \mathfrak h_B}~{}_{A'}\<\chi_b|\chi_b\>_{A'} \right)~.
\fe

Now, the overlap between $|\psi\>$ and $|\chi\>$ is
\ie
\<\psi|\chi\> &= {}_{A'}\<\Plus|\chi_+\>_{A'}~{}_{B'}\<\psi_+|\Plus\>_{B'}
\\
&\quad + \frac{1}{2^{|V_{A\cap B}|/2}} \sum_{a\in \ker \mathfrak h_A} ~{}_{A'}\<a_{|A'}|\chi_+\>_{A'}~{}_{B'}\<\psi_a|\Plus\>_{B'}
\\
&\quad + \frac{1}{2^{|V_{A\cap B}|/2}} \sum_{b\in \ker \mathfrak h_B} ~{}_{A'}\<\Plus|\chi_b\>_{A'}~{}_{B'}\<\psi_+|b_{|B'}\>_{B'}
\\
&\quad + \sum_{a\in \ker \mathfrak h_A, b\in \ker \mathfrak h_B \atop a_{|A\cap B} = b_{|A \cap B}} ~{}_{A'}\<a_{|A'}|\chi_b\>_{A'}~{}_{B'}\<\psi_a|b_{|B'}\>_{B'}~,
\fe
where we used \eqref{graph-innerproduct}. Using \eqref{graph-factorisation3}, the last line can be written as
\ie
\sum_{a\in \ker \mathfrak h_A, b\in \ker \mathfrak h_B \atop a_{|A\cap B} = b_{|A \cap B}} ~{}_{A'}\<a_{|A'}|\chi_b\>_{A'}~{}_{B'}\<\psi_a|b_{|B'}\>_{B'} = \sum_{c \in \ker \mathfrak h_{A\cup B}} ~{}_{A'}\<c_{|A'}|\chi_{c_{|B}}\>_{A'}~{}_{B'}\<\psi_{c_{|A}}|c_{|B'}\>_{B'}~.
\fe
Then, using the constraints \eqref{graph-ortho1}, \eqref{graph-ortho2}, and \eqref{graph-ortho3}, we have
\ie
\<\psi|\chi\> &= \frac{1}{2^{(|V_A|+|V_B|)/2}} \sum_{a\in \ker \mathfrak h_A, b\in \ker \mathfrak h_B}~{}_{A'}\<\Plus|\chi_b\>_{A'}~{}_{B'}\<\psi_a|\Plus\>_{B'}
\\
&\quad + \frac{1}{2^{|V_{A\cap B}|/2}} \sum_{a\in \ker \mathfrak h_A} ~{}_{A'}\<a_{|A'}|\chi_+\>_{A'}~{}_{B'}\<\psi_a|\Plus\>_{B'}
\\
&\quad + \frac{1}{2^{|V_{A\cap B}|/2}} \sum_{b\in \ker \mathfrak h_B} ~{}_{A'}\<\Plus|\chi_b\>_{A'}~{}_{B'}\<\psi_+|b_{|B'}\>_{B'}
\\
&\quad + \frac{1}{2^{(|V_A|+|V_B|)/2}} \sum_{c \in \ker \mathfrak h_{A\cup B}} ~{}_{A'}\<c_{|A'}|\chi_+\>_{A'}~{}_{B'}\<\psi_+|c_{|B'}\>_{B'}~.
\fe
Taking the absolute value gives
\ie
|\<\psi|\chi\>| &\le \frac{1}{2^{(|V_A|+|V_B|)/2}} \sum_{a\in \ker \mathfrak h_A, b\in \ker \mathfrak h_B}~|{}_{A'}\<\Plus|\chi_b\>_{A'}|~|{}_{B'}\<\psi_a|\Plus\>_{B'}|
\\
&\quad + \frac{1}{2^{|V_{A\cap B}|/2}} \sum_{a\in \ker \mathfrak h_A} ~|{}_{A'}\<a_{|A'}|\chi_+\>_{A'}|~|{}_{B'}\<\psi_a|\Plus\>_{B'}|
\\
&\quad + \frac{1}{2^{|V_{A\cap B}|/2}} \sum_{b\in \ker \mathfrak h_B} ~|{}_{A'}\<\Plus|\chi_b\>_{A'}|~|{}_{B'}\<\psi_+|b_{|B'}\>_{B'}|
\\
&\quad + \frac{1}{2^{(|V_A|+|V_B|)/2}} \sum_{c \in \ker \mathfrak h_{A\cup B}} ~|{}_{A'}\<c_{|A'}|\chi_+\>_{A'}|~|{}_{B'}\<\psi_+|c_{|B'}\>_{B'}|~.
\fe
Let us take care of each line separately using Cauchy-Schwarz inequality and the norm constraints \eqref{graph-norm1} and \eqref{graph-norm2}. First, we have
\ie
&\sum_{a\in \ker \mathfrak h_A, b\in \ker \mathfrak h_B}~|{}_{A'}\<\Plus|\chi_b\>_{A'}|~|{}_{B'}\<\psi_a|\Plus\>_{B'}|
\\
&\le \sqrt{2^{\nu_B} \sum_{b\in \ker \mathfrak h_B} |{}_{A'}\<\Plus|\chi_b\>_{A'}|^2}~\sqrt{2^{\nu_A} \sum_{a\in \ker \mathfrak h_A} |{}_{B'}\<\psi_a|\Plus\>_{B'}|^2}
\\
&\le 2^{(\nu_A+\nu_B)/2} \sqrt{\sum_{b\in \ker \mathfrak h_B} {}_{A'}\<\chi_b|\chi_b\>_{A'}}~\sqrt{\sum_{a\in \ker \mathfrak h_A} {}_{B'}\<\psi_a|\psi_a\>_{B'}}
\\
&\le \frac{2^{(\nu_A+\nu_B)/2}}{\sqrt{\left( 1-\frac{2^{\nu_B/2}}{2^{|V_B|/2}} \right) \left( 1-\frac{2^{\nu_A/2}}{2^{|V_A|/2}} \right)}}~.
\fe
Next, we have
\ie
&\sum_{a\in \ker \mathfrak h_A} ~|{}_{A'}\<a_{|A'}|\chi_+\>_{A'}|~|{}_{B'}\<\psi_a|\Plus\>_{B'}|
\\
&\le \sqrt{\sum_{a\in \ker \mathfrak h_A} |{}_{A'}\<a_{|A'}|\chi_+\>_{A'}|^2}~\sqrt{\sum_{a\in \ker \mathfrak h_A} |{}_{B'}\<\psi_a|\Plus\>_{B'}|^2}
\\
&\le \sqrt{2^{\nu_{A\cap B}} \sum_{a' \in \{0,1\}^{|V_{A'}|}} |{}_{A'}\<a'|\chi_+\>_{A'}|^2}~\sqrt{\sum_{a\in \ker \mathfrak h_A} |{}_{B'}\<\psi_a|\Plus\>_{B'}|^2}
\\
&\le 2^{\nu_{A\cap B}/2} \sqrt{{}_{A'}\<\chi_+|\chi_+\>_{A'}}~\sqrt{\sum_{a\in \ker \mathfrak h_A} {}_{B'}\<\psi_a|\psi_a\>_{B'}}
\\
&\le \frac{2^{\nu_{A\cap B}/2}}{\sqrt{\left( 1-\frac{2^{\nu_B/2}}{2^{|V_B|/2}} \right) \left( 1-\frac{2^{\nu_A/2}}{2^{|V_A|/2}} \right)}}~,
\fe
where, in the second inequality, we used the fact that if $a_1,a_2\in \ker\mathfrak h_A$ satisfy $a_{1|A'}=a_{2|A'}$, then $a_1+a_2 \in \ker \mathfrak h_{A\cap B}$. Similarly, we have
\ie
\sum_{b\in \ker \mathfrak h_B} ~|{}_{A'}\<\Plus|\chi_b\>_{A'}|~|{}_{B'}\<\psi_+|b_{|B'}\>_{B'}| \le \frac{2^{\nu_{A\cap B}/2}}{\sqrt{\left( 1-\frac{2^{\nu_B/2}}{2^{|V_B|/2}} \right) \left( 1-\frac{2^{\nu_A/2}}{2^{|V_A|/2}} \right)}}~.
\fe
And finally, we have
\ie
&\sum_{c \in \ker \mathfrak h_{A\cup B}} ~|{}_{A'}\<c_{|A'}|\chi_+\>_{A'}|~|{}_{B'}\<\psi_+|c_{|B'}\>_{B'}|
\\
&\le 2^{\nu_{A\cap B}} \sum_{a'\in \{0,1\}^{|V_{A'}|},b'\in \{0,1\}^{|V_{B'}|}} |{}_{A'}\<a'|\chi_+\>_{A'}|~|{}_{B'}\<\psi_+|b'\>_{B'}|
\\
&\le 2^{\nu_{A\cap B}} \sqrt{2^{|V_{A'}|} \sum_{a'\in \{0,1\}^{|V_{A'}|}} |{}_{A'}\<a'|\chi_+\>_{A'}|^2}~\sqrt{2^{|V_{B'}|} \sum_{b'\in \{0,1\}^{|V_{B'}|}} |{}_{B'}\<\psi_+|b'\>_{B'}|^2}
\\
&= 2^{\nu_{A\cap B} + (|V_{A'}|+|V_{B'}|)/2} \sqrt{{}_{A'}\<\chi_+|\chi_+\>_{A'}}~\sqrt{{}_{B'}\<\psi_+|\psi_+\>_{B'}}
\\
&\le \frac{2^{\nu_{A\cap B} + (|V_{A'}|+|V_{B'}|)/2}}{\sqrt{\left( 1-\frac{2^{\nu_B/2}}{2^{|V_B|/2}} \right) \left( 1-\frac{2^{\nu_A/2}}{2^{|V_A|/2}} \right)}}~.
\fe
Combining these inequalities, we have
\ie
|\<\psi|\chi\>| &\le \frac{1}{\sqrt{\left( 1-\frac{2^{\nu_B/2}}{2^{|V_B|/2}} \right) \left( 1-\frac{2^{\nu_A/2}}{2^{|V_A|/2}} \right)}} \left( \frac{2^{(\nu_A+\nu_B)/2}}{2^{(|V_A|+|V_B|)/2}} + \frac{2\times 2^{\nu_{A\cap B}/2}}{2^{|V_{A\cap B}|/2}} + \frac{2^{\nu_{A\cap B}}}{2^{|V_{A\cap B}|}} \right)~,
\fe
where we used $|V_A| = |V_{A'}| + |V_{A\cap B}|$ and $|V_B| =  |V_{A\cap B}| + |V_{B'}|$ to simplify the last term in the sum.

To simplify this inequality further, we need an upper bound on $\nu_A$ in terms of $|V_A|$, $D$, and $\widehat D$:
\ie
\nu_A \le \left(1-\frac1{2\Delta}\right)|V|~,
\fe
where $\Delta = \max(D,\widehat D)$. For a proof of this inequality, see Appendix \ref{app:nu-ineq}. Using this upper bound, the inequality simplifies to
\ie
|\<\psi|\chi\>| \le \frac{1}{\sqrt{\left( 1-\frac{1}{2^{|V_B|/4\Delta}} \right) \left( 1-\frac{1}{2^{|V_A|/4\Delta}} \right)}} \left( \frac{1}{2^{(|V_A|+|V_B|)/4\Delta}} + \frac{1}{2^{|V_{A\cap B}|/4\Delta}} + \frac{1}{2^{|V_{A\cap B}|/2\Delta}} \right)~.
\fe
Since $|V_A|,|V_B| \ge 2D \widehat D \ge 2\Delta$, we can further simplify this inequality to
\ie
|\<\psi|\chi\>| \le \frac{3\sqrt 2}{\sqrt 2-1} \times 2^{-|V_{A\cap B}|/4\Delta}~,
\fe
where we used $|V_A|,|V_B| \ge |V_{A\cap B}|$ to simplify the first term in the sum, and then used the fact that $\epsilon^2 < \epsilon$ for $0<\epsilon<1$ on the first and the third terms. Since this inequality holds for all $|\psi\>$ and $|\chi\>$ satisfying the constraints in \eqref{graph-delta-state-def}, we get the upper bound \eqref{graph-delta-upper-bound} on $\delta(A,B)$ for all sufficiently large $A$ and $B$.

\subsection{Proof of $\Pi_A \wedge \Pi_B = \Pi_{A\cup B}$}\label{app:proj-wedge-union}
This is equivalent to the statement $\mathcal H^0_{A\cup B} = \mathcal H^0_A \cap \mathcal H^0_B$. The direction $\mathcal H^0_{A\cup B} \subseteq \mathcal H^0_A \cap \mathcal H^0_B$ follows from frustration-freeness. Let us prove the other direction.

Consider a state $|\psi\> \in \mathcal H^0_A \cap \mathcal H^0_B$. It satisfies $P_v Q_{\hat v} |\psi\> = 0$ for all $v,\hat v \in A$ with $d(v,\hat v) = 3$ and for all $v,\hat v \in B$ with $d(v,\hat v) = 3$. But for this state to be contained in $\mathcal H^0_{A\cup B}$, the above constraint should hold for all $v,\hat v \in A\cup B$ with $d(v,\hat v) = 3$. The subtlety here is that it is possible to have $v\in A'$ and $\hat v\in B'$, or vice versa. Then, the constraint $P_v Q_{\hat v} |\psi\> = 0$ does not follow immediately. Nonetheless, we proceed as follows.

Our argument here is very similar to the one in Appendix \ref{app:eq-wt}. First, we write
\ie
|\psi\> = \sum_{\sigma\in \{+,1\}^{|V_{A\cup B}|}} \psi_\sigma |\sigma\>_{A\cup B}~.
\fe
where $|\sigma\>_{A\cup B}$ is an X-state and $\psi_\sigma$ is the ``wave function'' in the ``sign-configuration space''. Let $v\in A'$ and $\hat v \in B'$ with $d(v,\hat v) = 3$. (The case where $v\in B'$ and $\hat v \in A'$ with $d(v,\hat v) = 3$ is similar.) Then, the constraint $P_v Q_{\hat v} |\psi\> = 0$ is satisfied if (and only if) the following is true: for any $|\sigma\>_{A\cup B} \ne |\Plus\>_{A\cup B}$, if $\sigma_v = -$, then $\psi_{\sigma'} = \psi_{\sigma}$, where $|\sigma'\>_{A\cup B} = F_{\hat v} |\sigma\>_{A\cup B}$. In other words, we want to show that the flip $F_{\hat v}$ can always be ``simulated'' using only flips contained entirely in $A$ or in $B$.

Pick a vertex $\hat w \in \widehat V_{A\cap B}$ (recall that $\widehat V_{A\cap B} \ne \emptyset$ by assumption). Since $A$ and $B$ are connected, there are two paths, one from $v$ to $\hat w$ within $A$ and another from $\hat w$ to $\hat v$ within $B$. This yields a path in $A\cup B$ from $v$ to $\hat v$ via $\hat w$. One can apply the sequence of special flips described in scenario 1 of Appendix \ref{app:eq-wt} along this path such that each special flip is contained entirely in $A$ or in $B$. Therefore, we have simulated the flip $F_{\hat v}$.

\subsection{Upper bound on $\nu_A$}\label{app:nu-ineq}
For any (not necessarily bipartite) graph $\mathcal G$ with vertex-set $\mathcal V$, adjacency matrix $\mathfrak a$, and maximum degree at most $\Delta$, we have the inequality
\ie\label{nu-a-ineq}
\dim \ker \mathfrak a \le \left(1 - \frac1\Delta \right) |\mathcal V|~.
\fe
or equivalently,
\ie
\rank \mathfrak a \ge \frac{|\mathcal V|}\Delta~.
\fe
One can derive this inequality from the following significantly general inequality:
\begin{theorem}[{{\!\!\cite[Proposition 2.4]{MinrankZeroforcing}}}]\label{thm:graph-nullity}
Let $\mathbb F$ be any field and $\mathcal G = (\mathcal V,\mathcal E)$ be a simple undirected graph. Let $S$ be any symmetric square matrix of size $|\mathcal V|$ with entries in $\mathbb F$ such that $S_{xy} \ne 0$ if $\{x,y\} \in \mathcal E$ and $0$ otherwise. Then,
\ie
\dim_{\mathbb F} \ker_{\mathbb F} S \le Z(\mathcal G)~,
\fe
where $Z(\mathcal G)$ is the \emph{zero forcing number} of $\mathcal G$.
\end{theorem}
\noindent For our purposes, we do not need to know how $Z(\mathcal G)$ is defined, but what is important is that it depends only on the graph $\mathcal G$, i.e., it is independent of the matrix $S$ and the base field $\mathbb F$. In fact,
\begin{theorem}[{{\!\!\cite[Proposition 2]{GENTNER2018203}}}]\label{thm:graph-zero-forcing}
For any bounded-degree graph, where $\mdeg(\mathcal G) \le \Delta$, its zero forcing number satisfies\footnote{Actually, we need additionally that $\Delta \ge 3$ and $\mathcal G$ is not the complete graph $K_{\Delta+1}$. The former can be satisfied trivially by simply stating that $\Delta \ge 3$ and the latter is never encountered because we are interested in the case where $|\mathcal V|$ is increasing while $\Delta$ is held fixed.}
\ie
Z(\mathcal G) \le \left(1 - \frac1\Delta \right) |\mathcal V|~.
\fe
\end{theorem}
\noindent The bound \eqref{nu-a-ineq} follows from these two facts.

Coming back to our application, where $\mathcal G$ is a bipartite graph with vertex-set $\mathcal V = \widehat V \sqcup V$, the adjacency matrix $\mathfrak a$ takes the form \eqref{graph-adj-matrix}, where $\mathfrak h$ is the biadjacency matrix. So $\rank \mathfrak a = 2 \rank \mathfrak h$. It follows that $\rank \mathfrak h \ge |\mathcal V|/2\Delta$, which implies
\ie\label{nu-ineq}
\nu = \dim \ker \mathfrak h = |V| - \rank \mathfrak h \le |V| - \frac{|\mathcal V|}{2\Delta} \le \left(1-\frac1{2\Delta}\right)|V|~,
\fe
where $\Delta = \max(D,\widehat D)$.

One can generalise the above result from qubits to $\mathbb Z_N$ qudits easily. Indeed, when $N$ is prime, $\mathbb Z_N$ is a field, so the above argument goes through and we arrive at the same inequality \eqref{nu-ineq}. On the other hand, when $N$ is composite, one needs to be more careful.

First, there is no well defined notion of ``vector space'' and ``dimension'' in this case. Instead, by ``$\ker_{\mathbb Z_N} \mathfrak h$'', we simply mean the set of all $|V| \times 1$ column vectors that are annihilated by $\mathfrak h$ in $\mathbb Z_N$. The subtlety is that $|\ker_{\mathbb Z_N} \mathfrak h|$ is not necessarily a power of $N$, i.e., there is not always a $\nu$ such that $|\ker_{\mathbb Z_N} \mathfrak h| = N^\nu$. Fortunately, this does not mean we have to redo everything in Appendix \ref{sec:graph-delta}; all the of that discussion goes through except for the replacement
\ie
2^\nu = |\ker_{\mathbb Z_2} \mathfrak h| \to|\ker_{\mathbb Z_N} \mathfrak h|~,
\fe
which is well defined for all $N$. The desired generalisation of \eqref{nu-ineq} is, therefore,
\ie
|\ker_{\mathbb Z_N} \mathfrak h| \le N^{\left(1-\frac1{2\Delta}\right)|V|}~.
\fe
This inequality follows from the following result.
\begin{proposition}
Let $S$ be an $n\times n$ symmetric square matrix with entries in $\{0,1,\ldots,N-1\} \cong \mathbb Z_N$. Let $\mathcal G_S$ be the graph on $n$ vertices such that there is an edge between $x$ and $y$ if and only if $S_{x,y} \ne 0$. Say there is a positive integer $\Delta$ such that $\mdeg(\mathcal G_S) \le \Delta$, i.e., the number of non-zero entries in each row (or column) of $S$ is at most $\Delta$. Then,
\ie
|\ker_{\mathbb Z_N} S| \le N^{\left( 1-\frac1\Delta \right) n}~.
\fe
\end{proposition}
\begin{proof}
When $N$ is prime, the desired inequality follows directly from Theorems \ref{thm:graph-nullity} and \ref{thm:graph-zero-forcing}. So let us focus on composite $N$.

Consider $S$ with the same entries, but as a matrix over $\mathbb Z_p$ for any prime $p$, and let $\mathcal G^p_S$ be the corresponding graph. It is clear that $\mdeg(\mathcal G^p_S) \le \mdeg(\mathcal G_S)$. Moreover, $\mathbb Z_p$ is a field, so we can use Theorems \ref{thm:graph-nullity} and \ref{thm:graph-zero-forcing} to get
\ie
\dim_{\mathbb Z_p} \ker_{\mathbb Z_p} S \le Z(\mathcal G^p_S) \le \left( 1-\frac1\Delta \right) n~,
\fe
or equivalently,
\ie
\rank_{\mathbb Z_p} S := \dim_{\mathbb Z_p} \im_{\mathbb Z_p} S \ge \frac n\Delta~.
\fe
Since this inequality is true for any prime $p$, the number of $1$'s in the \emph{Smith normal form} of $S$ is at least $n/\Delta$. But then, over $\mathbb Z_N$, we have
\ie
|\im_{\mathbb Z_N} S| \ge N^{n/\Delta}~,
\fe
which is equivalent to the desired inequality.
\end{proof}

\bibliographystyle{JHEP}
\bibliography{Deformed_graph}

\end{document}